\documentclass[11pt, english]{article}
\usepackage{babel}
\usepackage[latin1]{inputenc}
\usepackage[T1]{fontenc}
\usepackage{amsmath}
\usepackage{graphicx}
\usepackage{amssymb}
\usepackage{amsthm}
\usepackage{subcaption}
\usepackage{bbm}
\usepackage{mathtools}

\usepackage{algorithmic}
\usepackage{courier} 
\usepackage{algorithm}
\graphicspath{{figures/}}

\usepackage[round]{natbib}

\usepackage{chngcntr}
\usepackage{apptools}
\AtAppendix{\counterwithin{thm}{section}}

\def\^{\widehat}

\newcommand{\norm}[1]{\Vert #1 \Vert}

\def\phi{\varphi}

\numberwithin{equation}{section}

\renewcommand{\phi}{\varphi}

\def\~{\widetilde}
\def\^{\widehat}
\newcommand{\ee}{{\rm e}\hspace{1pt}}

\newcommand{\dd}{\hspace{1pt}{\rm d}\hspace{0.5pt}}
\newcommand{\abs}[1]{\left| #1 \right|}

\newcommand{\deltas}[1]{\delta_{#1}(s)}
\newcommand{\deltat}[1]{\delta_{#1}(t)}

\newcommand{\veps}{\varepsilon}

\newtheorem{thm}{Theorem}

\newtheorem{lem}[thm]{Lemma}
\newtheorem{cor}[thm]{Corollary}
\newtheorem{defn}[thm]{Definition}

\newtheorem{remark}[thm]{\textit{Remark}}

\title{Computing Differential Privacy Guarantees for \\ Heterogeneous Compositions Using FFT}

\author{Antti Koskela and Antti Honkela \\ 
Helsinki Institute for Information Technology HIIT,\\
Department of Computer Science, University of Helsinki, Finland }
  
\date{}

\begin{document}
	
\maketitle

\abstract{

The recently proposed Fast Fourier Transform (FFT)-based accountant for evaluating $(\varepsilon,\delta)$-differential privacy guarantees using the privacy loss distribution formalism has been shown to give tighter bounds than commonly used methods such as R\'enyi accountants when applied to homogeneous compositions, i.e., to compositions of identical mechanisms. In this paper, we extend this approach to heterogeneous compositions. We carry out a full error analysis that allows choosing the parameters of the algorithm such that a desired accuracy is obtained. The analysis also extends previous results by taking into account all the parameters of the algorithm. Using the error analysis, we also give a bound for the computational complexity in terms of the error which is analogous to and slightly tightens the one given by Murtagh and Vadhan (2018). We also show how to speed up the evaluation of tight privacy guarantees using the Plancherel theorem at the cost of increased pre-computation and memory usage.

}

\section{Introduction}

Differential privacy (DP)~\citep{dwork_et_al_2006} has become
the standard approach for privacy-preserving machine learning. 
When using DP, one challenge is to accurately bound the
compound privacy loss of the increasingly complex DP algorithms. 
%
An important example is given by the differentially private stochastic gradient descent (DP-SGD).
The moments accountant~\citep{Abadi2016} gave a
major improvement in bounding the the complete $(\varepsilon,\delta)$-profile of the DP-SGD algorithm, and
this analysis has been refined through the general development of R\'enyi differential
privacy (RDP)~\citep{mironov2017} as well as tighter RDP bounds for
subsampled mechanisms~\citep{balle2018subsampling,wang2019,zhu2019,mironov2019} and improved conversion formulas~\cite{asoodeh2020}.
RDP enables nearly tight analysis for compositions of Gaussian
mechanisms, but this may be difficult for other mechanisms.
Moreover, conversion of RDP guarantees back to more commonly used
$(\varepsilon, \delta)$-guarantees is lossy.

In this work we use the privacy loss distribution (PLD)
formalism introduced by~\citet{sommer2019privacy} to numerically
evaluate tight privacy bounds for compositions. This work extends the
recent Fourier Accountant (FA) by~\citet{koskela2020,koskela2021tight} to heterogeneous compositions. 
This enables combining accurate accounting with more flexible algorithm design with non-uniform privacy budget spending.
Our work directly builds upon discrete mechanisms that are needed practical computer implementations of rigorous DP~\citep{mironov2012}.

FA uses numerical methods to compute very accurate privacy bounds. By taking into account the error analysis, these yield very tight rigorous bounds, but they cannot be expressed in a mathematically simple form. This appears to be a feature of bounds for complex mechanisms, as mathematically simple expressions provide either only approximate or very loose bounds.

\textbf{Our Contribution.} The main contributions of this work are:

\vspace{-0mm} 
 
\begin{itemize}
\item We extend the recently proposed FFT-based privacy accountant~\citet{koskela2020,koskela2021tight} for computing tight privacy bounds for heterogeneous compositions.

\item  We give a full error analysis for the method in terms of the pre-defined parameters of the algorithm. This also leads to strict upper $\delta(\veps)$-bounds.
We show that these bounds are accurate in a sense that they allow choosing appropriate parameter values a priori. 
We tailor the existing error analysis by~\citet{koskela2021tight} to heterogeneous compositions and extend it by analysing also the grid approximation error.


\item  Using the error analysis, we bound the computational complexity of the algorithm in terms of number of compositions $k$ and the tolerated error $\eta$. Our bound is slightly better than the 
existing bound given by~\citet{murtagh2018complexity}.

\item We show how to speed up the evaluation of the privacy parameters for varying numbers of compositions using the Plancherel theorem. 

\end{itemize}

\section{Differential Privacy}

We first recall some basic definitions of DP~\citep{dwork_et_al_2006}. We use the
following notation. An input data set containing $N$ data points is denoted as $X = (x_1,\ldots,x_N)
\in \mathcal{X}^N$, where $x_i \in \mathcal{X}$, $1 \leq i \leq N $.

\begin{defn} \label{def:adjacency}
	We say two data sets $X$ and $Y$ are neighbours in remove/add relation if you get 
	one by removing/adding an element from/to the other and denote this with $\sim_R$.
	We say $X$ and $Y$ are neighbours in substitute relation if you get one by substituting
	one element in the other. We denote this with $\sim_S$.
\end{defn}

\begin{defn} \label{def:dp}
	Let $\varepsilon > 0$ and $\delta \in [0,1]$. Let $\sim$ define a neighbouring relation.
	Mechanism $\mathcal{M} \, : \, \mathcal{X}^N \rightarrow \mathcal{R}$ is $(\varepsilon,\delta,\sim)$-DP 
	if for every $X \sim Y$
	and every measurable set $E \subset \mathcal{R}$ we have that
	$$
		\mathrm{Pr}( \mathcal{M}(X) \in E ) \leq \ee^\varepsilon \mathrm{Pr} (\mathcal{M}(Y) \in E ) + \delta.
	$$
	When the relation is clear from context or irrelevant, we will abbreviate it as $(\veps, \delta)$-DP. 
	We call $\mathcal{M}$ tightly $(\veps,\delta,\sim)$-DP, if there does not exist $\delta' < \delta$
	such that $\mathcal{M}$ is $(\veps,\delta',\sim)$-DP.
\end{defn}

\section{Privacy Loss Distribution} \label{sec:pld}

We first introduce the basic tool for obtaining tight privacy bounds: the privacy loss distribution (PLD).
The results in Subsection~\ref{subsec:pld} are reformulations of the results given by~\citet{meiser2018tight} and~\citet{sommer2019privacy}.
Proofs of the results of this section are given in the Supplements of~\citep{koskela2021tight}. 

\subsection{Privacy Loss Distribution}  \label{subsec:pld}

We consider discrete-valued one-dimensional mechanisms $\mathcal{M}$ which can be seen as mappings from $\mathcal{X}^N$
to the set of discrete-valued random variables.
The \emph{generalised probability density functions} of $\mathcal{M}(X)$ and $\mathcal{M}(Y)$, denoted $f_X(t)$ and $f_Y(t)$, respectively, are given by
\begin{equation} \label{eq:delta_sum}
	\begin{aligned}
		f_X(t) &= \sum\nolimits_i a_{X,i} \cdot \deltat{t_{X,i}},  \quad
		f_Y(t) &= \sum\nolimits_i a_{Y,i} \cdot \deltat{t_{Y,i}},
	\end{aligned}
\end{equation}
where $\delta_t( \cdot )$, $t \in \mathbb{R}$, 
denotes the Dirac delta function centred at $t$, and $t_{X,i},t_{Y,i} \in \mathbb{R}$ and $a_{X,i},a_{Y,i} \geq 0$.
We refer to~\citep{koskela2021tight} for more details of the notation.
The privacy loss distribution is defined as follows.
\begin{defn} \label{def:pld}
Let $\mathcal{M} \, : \, \mathcal{X}^N \rightarrow \mathcal{R}$, $\mathcal{R} \subset \mathbb{R}$, be a discrete-valued randomised mechanism
and let $f_X(t)$ and $f_Y(t)$ be generalised probability density functions of the form \eqref{eq:delta_sum}.
We define the generalised privacy loss distribution (PLD) $\omega_{X/Y}$ as   
\begin{equation} \label{eq:omega_pld}
	\begin{aligned}
	\omega_{X/Y}(s) &= \sum\nolimits_{{t_{X,i} = t_{Y,j} }}   a_{X,i} \cdot \deltas{s_i}, \quad s_i = \log \left( \tfrac{a_{X,i}}{a_{Y,j}} \right).
	\end{aligned}
\end{equation}
\end{defn}

\subsection{Tight $(\veps,\delta)$-Bounds for Compositions}
%
Let the generalised probability density functions $f_X$ and $f_Y$ of the form \eqref{eq:delta_sum}.
We define the convolution $f_X * f_Y$ as
\begin{equation*}
	(f_X * f_Y )(t) =  \sum\nolimits_{i,j} a_{X,i} \, a_{Y,j} \cdot \deltat{t_{X,i} + t_{Y,j}}.  
\end{equation*}
We consider non-adaptive compositions of the form $\mathcal{M}(X) = \big(\mathcal{M}_1(X), \ldots, \mathcal{M}_k(X) \big)$
and we denote by $f_{X,i}(t)$ the density function of $\mathcal{M}_i(X)$ for each $i$,
and by $f_{Y,i}(t)$ that of $\mathcal{M}_i(Y)$. 
For each $i$, 
we denote the PLD as defined by Definition~\ref{def:pld}
and the densities $f_{X,i}(t)$ and $f_{Y,i}(t)$ by $\omega_{X/Y,i}$.

The following theorem shows that the tight $(\veps,\delta)$-bounds for compositions 
of non-adaptive mechanisms are obtained using convolutions of PLDs (see also Thm.\;1 by~\citet{sommer2019privacy}).
\begin{thm} \label{thm:integral}
Consider a non-adaptive composition of $k$ independent mechanisms $\mathcal{M}_1,\ldots,\mathcal{M}_k$ and neighbouring data sets $X$ and $Y$.
The composition is tightly $(\veps,\delta)$-DP for $\delta(\veps)$ given by
$$
\delta(\veps) = \max \{ \delta_{X/Y}(\veps), \delta_{Y/X}(\veps) \},
$$ 
where
\begin{equation*} 
	\begin{aligned}
		 \delta_{X/Y}(\veps) = 1 - \prod\limits_{\ell=1}^k (1-\delta_{X/Y,\ell}(\infty)) 
	+ \int_\veps^\infty (1 - \ee^{\veps - s})\left(\omega_{X/Y,1} * \cdots * \omega_{X/Y,k} \right) (s)  \, \dd s,
	\end{aligned}
\end{equation*} 
\begin{equation} 
	\begin{aligned}
 \delta_{X/Y,\ell}(\infty) = 
 \sum\limits_{ \{ t_i \, : \, \mathbb{P}( \mathcal{M_\ell}(X) = t_i) > 0, \, \mathbb{P}( \mathcal{M_\ell}(Y) = t_i) = 0 \} }
 		\mathbb{P}( \mathcal{M_\ell}(X) = t_i)
	\end{aligned}
\end{equation}
and $\omega_{X/Y,1} * \cdots * \omega_{X/Y,k}$ denotes the convolution of 
the density functions $\omega_{X/Y,\ell}$, $1 \leq \ell \leq k$. An analogous expression holds for $\delta_{Y/X}(\veps)$.
\end{thm}

We remark that finding the outputs $\mathcal{M}_i(X)$ and $\mathcal{M}_i(Y)$, $1 \leq i \leq k$, that give the maximum $\delta(\varepsilon)$ 
is application-specific and has to be carried out individually for each case, similarly as, e.g., in the case of 
RDP~\citep{mironov2017}. In the experiments of Section~\ref{sec:experiments} it will be clear how to determine the worst-case distributions $f_{X,i}$ and $f_{Y,i}$.

\section{Fourier Accountant for Heterogeneous Compositions}


We next describe the numerical method for computing tight DP guarantees for heterogeneous compositions of discrete-valued mechanisms.
The method is closely related to the homogenous case considered in~\citep{koskela2021tight}.
However, the error analysis is tailored to the heterogeneous case and we  consider here also the error induced by the grid approximation.

%
\subsection{Fast Fourier Transform}

We first recall some basics of the Fast Fourier Transform (FFT)~\citep{cooley1965}. Let
$$
x = \begin{bmatrix} x_0,\ldots,x_{n-1} \end{bmatrix}^\mathrm{T}, \, w =  \begin{bmatrix} w_0,\ldots,w_{n-1} \end{bmatrix}^\mathrm{T} \in \mathbb{R}^n.
$$
The discrete Fourier transform $\mathcal{F}$ and its inverse $\mathcal{F}^{-1}$ are defined as~\citep{stoer_book}
\begin{equation}  \label{def:mathcal{F}}
	\begin{aligned}
	(\mathcal{F} x)_k  = \sum\nolimits_{j=0}^{n-1} x_j \ee^{- \mathrm{i} \, 2 \pi k j / n}, \quad 
		(\mathcal{F}^{-1} w  )_k = \frac{1}{n} \sum\nolimits_{j=0}^{n-1} w_j \ee^{ \mathrm{i} \, 2 \pi k j / n},
	\end{aligned}
\end{equation}
where $\mathrm{i} = \sqrt{-1}$.
Evaluating $\mathcal{F} x$ and $\mathcal{F}^{-1}w$  naively takes $O(n^2)$ operations,
however by using FFT 
the running time complexity reduces to $O(n \log n)$.
Also, FFT enables evaluating discrete convolutions
efficiently. The convolution theorem~\citep{stockham1966} states that 
\begin{equation*} 
	\sum\nolimits_{i=0}^{n-1} v_i w_{k-i} = \mathcal{F}^{-1} ( \mathcal{F} v \odot \mathcal{F} w),
\end{equation*}
where $\odot$ denotes the element-wise product and the summation indices are  modulo $n$.

\subsection{Grid Approximation}


Similarly as~\citet{koskela2021tight}, we place the PLD on a grid
\begin{equation} \label{eq:grid}
X_n = \{x_0,\ldots,x_{n-1}\}, \quad n \in \mathbb{Z}^+, \quad \textrm{where} \quad x_i = -L + i \Delta x, \quad  \Delta x = 2L/n. 
\end{equation}
Suppose the distribution $\omega$ of the PLD is of the form 
\begin{equation} \label{eq:omega_0}
	\omega(s) = \sum\nolimits_i a_i \cdot \deltas{s_i},
\end{equation}
where $a_i \geq 0$ and $-L \leq s_i \leq L - \Delta x$ for all $i$. 
We define the grid approximations
\begin{equation} \label{eq:omegaRL}
	\begin{aligned}
		\omega^\mathrm{L}(s)  := \sum\nolimits_i a_i \cdot \deltas{s_i^\mathrm{L}}, \quad
		\omega^\mathrm{R}(s)  := \sum\nolimits_i a_i \cdot \deltas{s_i^\mathrm{R}}, 
	\end{aligned}
\end{equation}
where
\begin{equation*}
	\begin{aligned}
		   s_i^\mathrm{L} = \max \{  x \in X_n \, : \, x \leq s_i  \}, \quad 
	       s_i^\mathrm{R} = \min \{  x \in X_n \, : \, x \geq  s_i\}.
	\end{aligned}
\end{equation*}
We note that $s_i$'s correspond to the logarithmic ratios of probabilities of individual events.
Thus, often a moderate $L$ is sufficient for the condition $-L \leq s_i \leq L - \Delta x$
to hold for all $i$.
We also provide analysis for the case where this assumption does not hold (see the Appendix). 
From \eqref{eq:omegaRL} we get: 
\begin{lem} \label{lem:deltaineq}
Let $\delta(\veps)$ be given by the integral formula of Theorem~\ref{thm:integral} for PLDs $\omega_1, \cdots, \omega_k$ of the form \eqref{eq:omega_0}.
Let $\delta^\mathrm{L}(\veps)$ and $\delta^\mathrm{R}(\veps)$  correspondingly be determined by the left and right 
approximations $\omega_1^\mathrm{L}, \ldots,  \omega_k^\mathrm{L} $ and $\omega_1^\mathrm{R}, \ldots,  \omega_k^\mathrm{R} $, as defined in \eqref{eq:omegaRL}.
Then for all $\veps>0$ :
\begin{equation*} 
	\delta^\mathrm{L}(\veps) \leq \delta(\veps) \leq \delta^\mathrm{R}(\veps).
\end{equation*}
\end{lem}

\subsection{Truncation and Periodisation} 

By truncating convolutions and periodising the PLD distributions we arrive at periodic sums to which the FFT is directly applicable.
These operations are analogous to the homogeneous case described in~\citep{koskela2021tight}. We describe them next shortly.

Suppose $\omega_1$ and $\omega_2$ are defined such that 
\begin{equation} \label{eq:omega}
	\omega_1(s) = \sum\nolimits_i a_i \cdot \deltas{s_i}, \quad \omega_2(s) = \sum\nolimits_i b_i \cdot \deltas{s_i},
\end{equation}
where for all $i$: $a_i,b_i \geq 0$ and $s_i = i \Delta x$. 
The convolution $\omega_1 * \omega_2$ can then be written as 
\begin{equation} \label{eq:o_conv}
	\begin{aligned}
	(\omega_1 * \omega_2)(s)  & =  \sum\nolimits_{i,j} a_i b_j \cdot \deltas{s_i + s_j} \quad
	 = \sum\nolimits_i \Big(\sum\nolimits_j a_j b_{i-j} \Big) \cdot \deltas{s_i}.
\end{aligned}
\end{equation}
Let $L>0$. We truncate convolutions to the interval $[-L,L]$: 
\begin{equation*} 
	\begin{aligned}
	(\omega_1 * \omega_2 )(s)  \approx \sum\nolimits_i \Big(\sum\nolimits_{-L \leq s_j < L} a_j b_{i-j} \Big) \cdot \deltas{s_i}   
	 =: (\omega_1 \circledast \omega_2 )(s).
\end{aligned}
\end{equation*}
For $\omega_1$ of the form \eqref{eq:omega}, we define $\widetilde{\omega}_1$ to be a $2 L$-periodic extension of $\omega_1$
from $[-L,L)$ to $\mathbb{R}$, i.e., $\widetilde{\omega}_1$ is of the form
\begin{equation*} 
	\widetilde{\omega}_1(s) = \sum\nolimits_{m \in \mathbb{Z}} \, \sum\nolimits_i a_i \cdot \deltas{s_i + m \cdot 2 L}.
\end{equation*}
For $\omega_1$ and $\omega_2$ of the form \eqref{eq:omega}, we approximate the convolution $\omega_1 * \omega_2$ as
\begin{equation} \label{eq:conv_approximation}
	\omega_1 * \omega_2 \approx \widetilde{\omega}_1 \circledast \widetilde{\omega}_2.
\end{equation}
Since $\omega_1$ and $\omega_2$ are defined on an equidistant grid, 
FFT can be used to evaluate the approximation $\widetilde{\omega}_1 \circledast \widetilde{\omega}_2$
as follows:
\begin{lem} \label{lem:fft}
Let $\omega_1$ and $\omega_2$ be of the form \eqref{eq:omega}, such that $s_i = -L + i \Delta x$, $0 \leq i \leq n-1$, where $L>0$, $n$ is even and 
$\Delta x = 2L/n$.
Define
\begin{equation*} 
 	\begin{aligned}
\boldsymbol{a}  = \begin{bmatrix} a_0 & \ldots & a_{n-1} \end{bmatrix}^\mathrm{T},
\quad  \boldsymbol{b} = \begin{bmatrix} b_0 & \ldots & b_{n-1} \end{bmatrix}^\mathrm{T},
\quad D = \begin{bsmallmatrix} 0 & I_{n/2} \\ I_{n/2} & 0 \end{bsmallmatrix} \in \mathbb{R}^{n \times n}.
	\end{aligned}
\end{equation*} 
Then, 
$$
(\widetilde{\omega}_1 \circledast \widetilde{\omega}_2 )(s) = \sum\nolimits_{i=0}^{n-1} c_i \cdot \deltas{s_i},
\quad
\textrm{where}
\quad
c_i = \left[D \, \mathcal{F}^{-1} \big(\mathcal{F}( D \, \boldsymbol{a} ) \odot \mathcal{F}( D \, \boldsymbol{b} )    \big) \right]_i.
$$
\end{lem}
Since the coefficients of $\widetilde{\omega}_1 \circledast \widetilde{\omega}_2$ are given by the discrete Fourier transform,
we are able to analyse the error induced by the FFT approximation by only considering the error of the approximation \eqref{eq:conv_approximation}.

\begin{algorithm}[ht!]
\caption{ Fourier Accountant Algorithm for Heterogeneous Discrete-Valued Mechanisms}
\begin{algorithmic}
\STATE{Input: distributions $\omega_1, \ldots, \omega_m$ of the form 
$\omega_j(s) = \sum\nolimits_i a_i^j \cdot \deltas{s_i}$,
$1 \leq j \leq m$, such that $s_i = -L + i \Delta x$, where 
$n$ is even and, $0 \leq i \leq n-1$, $\Delta x=2L/n$. 
Numbers of compositions for each mechanism, $k_1,\ldots, k_m$.}
\vspace{2mm}
\STATE{Set
\begin{equation*}
\begin{aligned}
	\boldsymbol{a}^j = \begin{bmatrix} a_0^j & \ldots & a_{n-1}^j \end{bmatrix}^\mathrm{T}, \quad 1 \leq j \leq m.
\end{aligned}
\end{equation*}
}
\STATE{ 
For each $j$, $1 \leq j \leq m$, evaluate the FFT:
$$
\widetilde{\boldsymbol{a}}^j = \mathcal{F}( D \boldsymbol{a}^j ).
$$
Compute the element-wise products and apply $\mathcal{F}^{-1}$: 
\begin{equation*}
\boldsymbol{b} = \left[D \, \mathcal{F}^{-1} \big( (\widetilde{\boldsymbol{a}}^1)^{\odot k_1} \odot \cdots \odot (\widetilde{\boldsymbol{a}}^m)^{\odot k_m} \big) \right].
\end{equation*}
} 
\STATE{Approximate $\delta(\veps)$: 
\begin{equation*}
	\begin{aligned}
		\delta( \veps)  \approx  1 - \prod\limits_{\ell=1}^m (1-\delta_{X/Y,\ell}(\infty))^{k_\ell} 
	 +  \sum\limits_{\{ \ell \, : \, -L + \ell \Delta x > \veps \}}  \big(1 - \ee^{\veps - ( - L + \ell \Delta x)} \big) \, b_\ell,
	\end{aligned}
\end{equation*}
where $\delta_{X/Y,\ell}(\infty)$ is defined in Theorem~\ref{thm:integral}.
}
\end{algorithmic}
\label{alg:delta}
\end{algorithm}

\subsection{Computing Upper Bounds for $\delta(\veps)$}

Given a discrete-valued PLD distribution $\omega$, we get a strict upper $\delta(\veps)$-DP bound as follows.
Using parameter values $L>0$ and $n \in \mathbb{Z}^+$, we form a grid $X_n$ as defined in \eqref{eq:grid} and 
place each PLDs $\omega_i$, $1 \leq i \leq k$, on $X_n$ to obtain $\omega_i^\mathrm{R}$ as defined in \eqref{eq:omegaRL}. 
We then approximate $\delta^R(\veps)$ using Algorithm~\ref{alg:delta}. We estimate the error incurred by truncation of convolutions
periodisation of PLDs using  Thm.~\ref{thm:alg_error_bound1}.
By adding this error to the approximation given by Algorithm~\ref{alg:delta} 
we obtain a strict upper bound for $\delta(\veps)$.
The parameter $n$ can be increased in case the discretisation error bound given by Thm.~\ref{thm:discretisation_error} is too large.


\section{Error Analysis} \label{sec:err_est}

We next bound the error induced by the grid approximation and Algorithm~\ref{alg:delta}. 
The total error consists of (see the Appendix for more details)
\begin{enumerate}
	\item The tail integral $\int_L^\infty (1 - \ee^{\veps - s})(\omega_1 * \cdots * \omega_k ) (s)  \, \dd s$.
	\item The error arising from periodisation of $\omega$ and truncation of the convolutions (affected by $L$):
	\begin{equation*}
		\begin{aligned}
			 \int_\veps^L (1 - \ee^{\veps - s})(\omega_1 * \cdots * \omega_k ) (s)  \, \dd s 
			-  \int_\veps^L (1 - \ee^{\veps - s})(\widetilde{\omega}_1 \circledast \cdots  \circledast \widetilde{\omega}_k ) (s)  \, \dd s.
		\end{aligned}
	\end{equation*}
	\item The discretisation error arising from the grid approximations (affected by both $L$ and $n$):
	\begin{equation*}
		\begin{aligned}
				 \int_\veps^L (1 - \ee^{\veps - s})(\omega_1 * \cdots * \omega_k ) (s)  \, \dd s 
				-  \int_\veps^L (1 - \ee^{\veps - s})(\omega_1^\mathrm{R} * \cdots * \omega_k^\mathrm{R}) (s)  \, \dd s.
		\end{aligned}
	\end{equation*}
\end{enumerate}

\subsection{Bounding Tails Using the Chernoff Bound} \label{subsubsec:num_approx1}

We obtain error bounds essentially using 
the Chernoff bound 
$\mathbb{P}[Z \geq t] 
\leq \tfrac{ \mathbb{E}[ \ee^{\lambda Z} ] }{\ee^{\lambda t}}$
which holds for any random variable $Z$ and all $\lambda > 0$.
Suppose $\omega_{X/Y}$ is of the form 
\begin{equation} \label{eq:omega_xy}
	\omega_{X/Y}(s) = \sum\nolimits_i a_{X,i} \cdot \deltas{s_i},  
\end{equation} 
where $s_i = \log \left( \tfrac{a_{X,i}}{a_{Y,i}} \right)$, $a_{X,i},a_{Y,i}>0$.  
Then, the moment generating function of $\omega_{X/Y}$ is given by
\begin{equation} \label{eq:pld_lmf}
	\begin{aligned}
		\mathbb{E} [\ee^{\lambda \omega_{X/Y} }] 
		 = \sum\nolimits_i  \ee^{\lambda s_i} \cdot a_{X,i}  
		 = \sum\nolimits_i  \left( \frac{a_{X,i}}{a_{Y,i}}  \right)^\lambda a_{X,i}.
	\end{aligned}
\end{equation}

In our analysis, we repeatedly use the Chernoff bound to bound tails of PLD distributions in terms of
pre-computable moment-generating functions.
Denote $S_k := \sum_{i=1}^k \omega_i$, where $\omega_i$ denotes the PLD random variable of the $i$th mechanism.
If $\omega_i$'s are independent, 
$
\mathbb{E} [ \ee^{\lambda S_k}  ]  = \prod\nolimits_{i=1}^k \mathbb{E} [ \ee^{\lambda \omega_i}  ].
$
Then, the Chernoff bound shows that for any $\lambda > 0$
\begin{equation} \label{eq:tail_bound}
	\begin{aligned}
		\int_L^\infty  ( \omega_{1} * \cdots * \omega_{k})(s) \, \dd s = \mathbb{P}[ S_k \geq L ]
		\leq \prod\nolimits_{i=1}^k \mathbb{E} [ \ee^{\lambda \omega_i}  ] \,  \ee^{- \lambda L}
		 \leq \ee^{\sum_{i=1}^k \alpha_i(\lambda)}  \ee^{- \lambda L},
	\end{aligned}
\end{equation}
where $\alpha_i(\lambda) = \log( \mathbb{E} [\ee^{\lambda \omega_i }] )$.

\subsection{Truncation and Periodisation Error} \label{sec:trunc_error}


Denote the logarithms of the moment generating functions of the PLDs as
$$
\alpha_i^+(\lambda) = \log \Big(	\mathbb{E} [\ee^{\lambda \omega_i }] \Big),
\quad  \alpha_i^-(\lambda) = \log \Big(	\mathbb{E} [\ee^{ - \lambda \omega_i }] \Big),
$$
where $1 \leq i \leq k$. Furthermore, denote 
\begin{equation} \label{eq:alphaplusminus}
	\alpha^+(\lambda) = \sum\nolimits_i \, \alpha^+_i(\lambda),
	\quad  \alpha^-(\lambda) = \sum\nolimits_i \, \alpha^-_i(\lambda).
\end{equation}
To obtain $\alpha^+(\lambda)$ and $\alpha^-(\lambda)$, 
we evaluate  the moment generating functions 
using the finite sum \eqref{eq:pld_lmf}. 

Using the analysis given in the Appendix, we bound the errors arising from the
periodisation of the distribution and truncation of the convolutions. As a result, when combining with the Chernoff bound \eqref{eq:tail_bound}, we obtain
the following two bounds for the total error incurred by Algorithm~\ref{alg:delta}.

\begin{thm} \label{thm:alg_error_bound1}
Let $\omega_i$'s be defined on the grid $X_n$ as described above (i.e., $s_j \in [-L,L-\Delta x]$ for all $j$). 
Let $\delta(\veps)$ give the tight $(\veps,\delta)$-bound for the PLDs $\omega_1, \ldots, \omega_k$
and let $\widetilde{\delta}(\veps)$ be the result of Algorithm~\ref{alg:delta}.
Then, for all $\lambda > 0$
\begin{equation*}
	\begin{aligned}
 \abs{ \delta(\veps)  -  \widetilde{\delta}(\veps) }
  \leq	 \big(  \ee^{\alpha^+(\lambda)} + \ee^{ \alpha^-(\lambda)} \big) \, \frac{\ee^{-L \lambda}}{1- \ee^{-2 L \lambda}}.
 	\end{aligned}
\end{equation*}	
\end{thm}
As $s_i$'s correspond to the logarithmic ratios of probabilities of individual events,
often a moderate $L$ is sufficient for $-L \leq s_i \leq L - \Delta x$ to hold for all $i$.
In the Appendix, we give a bound which holds also 
in case $s_i$'s are not inside the interval $[-L,L)$.  


\subsection{Bound for the Discretisation Error} \label{sec:discr_error}

Let $\omega_1,\ldots, \omega_k$ be PLD distributions of the form \eqref{eq:omega_0}. 
For each $\ell$, denote the PLD as $\omega_\ell(s) = \sum\nolimits_i a_i^\ell \cdot \deltas{s_i^\ell}$
and the corresponding left and right grid approximation (defined in \eqref{eq:omegaRL}) as
$$
\omega^\mathrm{L}_\ell(s) = \sum\nolimits_i a_i^\ell \cdot \deltas{s_i^{\mathrm{L},\ell}}  \quad \textrm{and} \quad 
\omega^\mathrm{R}_\ell(s) = \sum\nolimits_i a_i^\ell \cdot \deltas{s_i^{\mathrm{R},\ell}} 
$$
and the tight $(\veps,\delta)$-bound corresponding to the PLDs $\omega^\mathrm{L}_1 * \cdots * \omega^\mathrm{L}_k$ 
and $\omega^\mathrm{R}_1 * \cdots * \omega^\mathrm{R}_k$ by $\delta^\mathrm{L}(\veps)$ and $\delta^\mathrm{R}(\veps)$. 
We have the following bound for the error arising from the right grid approximation.
\begin{thm} \label{thm:discretisation_error}
Let $\delta(\veps)$ denote the tight $(\veps,\delta)$-bound for the convolution PLD $\omega_1 * \cdots * \omega_k$.
The discretisation error $\delta^\mathrm{R}(\veps) - \delta(\veps)$ can be bounded as
\begin{equation} \label{eq:discretisation_error}
	\delta^\mathrm{R}(\veps) - \delta(\veps) \leq   k \Delta x \,  \big( \mathbb{P}( \omega_1 + \cdots + \omega_k \geq \veps ) - \delta(\veps) \big).
\end{equation}	
\end{thm}


\begin{remark} \label{remark:discretisation_error}
Theorem~\ref{thm:discretisation_error} instantly gives the bound
\begin{equation} \label{eq:deltax_bound}
	\delta^\mathrm{R}(\veps) - \delta(\veps) \leq k \Delta x \, \big( 1 - \delta(\veps) \big) \leq k \Delta x.
\end{equation}
On the other hand, the bound \eqref{eq:discretisation_error} and the Chernoff bound \eqref{eq:tail_bound} give
\begin{equation} \label{eq:discr_upper}
	\begin{aligned}
	\delta^\mathrm{R}(\veps) - \delta(\veps)  \leq k \Delta x \, \mathbb{P}( \omega_1 + \cdots + \omega_k \geq \veps ) 
	 \leq k \Delta x \, \ee^{\sum\nolimits_i \alpha_i(\lambda)} \ee^{- \lambda \veps}
	\end{aligned}
\end{equation}
which holds for any $\lambda > 0$. By choosing $\lambda$ appropriately, this leads to a considerably tighter a priori bound than \eqref{eq:deltax_bound}.
\end{remark}


\textbf{Experimental Illustration. } Tables 1 to 3 illustrate the discretisation error bound \eqref{eq:discr_upper}.  
We consider the one-dimensional binomial mechanism~\citep{agarwal2018}, where a binomially distributed 
noise $Z$ with parameters $n \in \mathbb{N}$ and $0<p<1$ is added to the output of a query $f$. Denoting the sensitivity
of $f$ by $\Delta$, tight $(\veps,\delta)$-bounds are obtained by considering the PLD $\omega_{X/Y}$ given by the distributions
$f_X$ and $f_Y$, where 
$$ 
f_X \sim  \Delta + \mathrm{Bin}( N,p) \quad \textrm{and} \quad f_Y \sim \mathrm{Bin}(N,p).
$$
We set $N=1000$, $p=0.5$, $\Delta=1$ and $L=5.0$.
In the numerical implementation 
we compute logarithmic probabilities using the digamma function and use those to evaluate 
the values of $\alpha^+(\lambda)$ and $\alpha^+(\lambda)$ required by the error bounds. 
For the upper bound \eqref{eq:discr_upper} we take the minimum of the bounds computed with $\lambda \in \{0.5L,1.0L,2.0L,3.0L,4.0L\}$.


\begin{table}[h!]
\begin{center}
\begin{tabular}{ccc}
 \hline
 $n$ &     error bound \eqref{eq:discr_upper} & $\delta(\veps)$ \\
 \hline
	$10^5$ &   $6.31 \cdot 10^{-6}$ & $ 2.37864 \cdot 10^{-5}$ \\
	$10^6$ &   $6.31 \cdot 10^{-7}$ & $ 2.35330 \cdot 10^{-5}$ \\
	$10^7$ &   $6.31 \cdot 10^{-8}$ & $ 2.35039 \cdot 10^{-5}$ \\
	$10^8$ &   $6.31 \cdot 10^{-9}$ & $ 2.35011 \cdot 10^{-5}$ \\
 \hline
 \end{tabular} 
 \caption{The error bound \eqref{eq:discr_upper} for different values of $n$ when $\veps=1.0$, $k=20$, and the corresponding $\delta(\veps)$-value.
We see that the bound is not far from the magnitude of the actual error. }
 \label{table:disc3}
 \begin{tabular}{ccc}
  \hline
  $\veps$ &     error bound \eqref{eq:discr_upper} & $\delta(\veps)$ \\
  \hline
 	0.7 &   $1.32 \cdot 10^{-6}$ & $8.62596 \cdot 10^{-4}$ \\
 	1.1 &   $1.79 \cdot 10^{-8}$ & $5.66127 \cdot 10^{-6}$ \\
 	1.5 &   $3.31 \cdot 10^{-11}$ & $6.03580 \cdot 10^{-9}$ \\
 	1.9 &   $8.36 \cdot 10^{-15}$ & $9.82392 \cdot 10^{-13}$ \\
  \hline
  \end{tabular}
  \caption{The error bound \eqref{eq:discr_upper} for different values of $\veps$ when $n=10^7$ and $k=20$ and the corresponding $\delta(\veps)$-value. 
  We see that the bound \eqref{eq:discr_upper} stays small in relation to $\delta(\veps)$ as $\delta$ decreases. }
  \label{table:disc2}
\end{center}
\end{table}
%
\subsection{Upper Bound for the Computational Complexity}


The results by~\citet{murtagh2018complexity} state that there is no algorithm for computing tight $(\veps,\delta)$-bounds that would have polynomial complexity in $k$, number of compositions.
However, Theorem 1.7 by~\citet{murtagh2018complexity} states that allowing a small error in the output, the bounds can be evaluated efficiently.
Assuming there are $m<k$ distinct mechanisms in the composition, using the error analysis of Sections~\ref{sec:trunc_error} and~\ref{sec:discr_error}, 
we obtain the following bound for the evaluation of tight $\delta$ as a function of $\veps$.
This slightly improves the the complexity result by~\citet{murtagh2018complexity}.

\begin{lem}
Consider a non-adaptive composition of the mechanisms $\mathcal{M}_1, \ldots, \mathcal{M}_k$ with corresponding worst-case pairs of distributions
$f_{X,i}$ and $f_{Y,i}$, $1 \leq i \leq k$. Suppose the sequence $\mathcal{M}_1, \ldots, \mathcal{M}_k$ consists of $m$ distinct mechanisms.
Then, it is possible to have an approximation of $\delta(\veps)$ within error less than $\eta$
with number of operations
$$
\mathcal{O}\left(\frac{2 m \cdot k^2 \cdot C_k}{\eta} \log  \frac{k^2 \cdot C_k}{\eta} \right),
$$
where 
$$
C_k = \max \{ \tfrac{1}{k} \sum_i D_{\infty}(f_{X,i} || f_{Y,i}),  \tfrac{1}{k} \sum_i D_{\infty}(f_{Y,i} || f_{X,i}) \},
\quad
D_{\infty}(f_X || f_Y) = \sup_{a_{Y,i} \neq 0} \log \frac{a_{X,i}}{a_{Y,i}},
$$
and the additional factor in the leading constant is the leading constant of the FFT algorithm.
\end{lem}
\subsection{Fast Evaluation Using the Plancherel Theorem}

When using Algorithm~\ref{alg:delta} to approximate $\delta(\veps)$, we need to evaluate the expression 
\begin{equation}\label{eq:summ} 
	\widetilde{\delta}(\veps) = \sum\nolimits_{ - L + \ell \Delta x > \veps}  \big(1 - \ee^{\veps - ( - L + \ell \Delta x)} \big) \, b^k_\ell, 
	\quad \textrm{where} \quad
\boldsymbol{b}^k = D \, \mathcal{F}^{-1} \big(\mathcal{F}( D \boldsymbol{a} )^{\odot k}   \big).
\end{equation}
When evaluating $\widetilde{\delta}(\veps)$ for different numbers of compositions $k$, 
we see that the inverse transform $\mathcal{F}^{-1}$ is the most expensive part as  the vector $\mathcal{F}( D \boldsymbol{a} )$ can be precomputed.
The following lemma shows that using the Plancherel theorem the updates of $\widetilde{\delta}(\veps)$ can actually be performed in $\mathcal{O}(n)$ time:
%
\begin{lem} \label{lem:plancherel}
Denote $\boldsymbol{w}_\veps \in \mathbb{R}^n$ such that
$(\boldsymbol{w}_\veps)_\ell = \max\{  1 - \ee^{\veps - ( - L + \ell \Delta x)} , 0\}$,
and let $\widetilde{\delta}(\veps)$ be given by \eqref{eq:summ}.
Then, we have that 
\begin{equation} \label{eq:planch0}
\widetilde{\delta}(\veps) = \frac{1}{n} \langle \mathcal{F}( D \boldsymbol{w}_\veps ) , \mathcal{F}( D \boldsymbol{a} )^{\odot k}   \rangle.
\end{equation}
\end{lem}
%


\textbf{Experimental Illustration. } 
Consider computing tight $\delta(\veps)$-bound for the subsampled Gaussian mechanism (see Section~\ref{subsec:subsampled}), for $q=0.02$ and $\sigma=2.0$.
We evaluate $\delta(\veps)$ after $k=100, 200, \ldots, 500$ compositions at $\veps=1.0$. Table~\ref{table:plancherel} illustrates the compute time for each update of $\delta(\veps)$,
using a) a pre-computed vector $\mathcal{F}( D \boldsymbol{a} )$, the inverse transform $\mathcal{F}^{-1}$ and the summation \eqref{eq:summ} and b)
 using pre-computed vectors $\mathcal{F}( D \boldsymbol{a} )^{\odot 100}$ and 
 $\mathcal{F}( D \boldsymbol{w}_\veps )$ and the inner product~\eqref{eq:planch0}. 

\begin{table}[h!]
\begin{center}
\begin{tabular}{cccc}
 \hline
 $n$ &  $t$ (ms)~\eqref{eq:summ} & $t$ (ms)~\eqref{eq:planch0} & $\delta(\veps)$ \\
 \hline
	\hspace{-1mm} $5 \cdot 10^4$ \hspace{-3mm} &     5.8 \hspace{-3mm}&  \hspace{-3mm}  0.18 \hspace{-3mm}& \hspace{-2mm} $2.900925\cdot 10^{-6}$ \\
	\hspace{-1mm} $1 \cdot 10^5$ \hspace{-3mm} &     12 \hspace{-3mm}&   \hspace{-3mm}  0.36 \hspace{-3mm}& \hspace{-1mm}$2.851835\cdot 10^{-6}$ \\
	\hspace{-1mm} $1 \cdot 10^6$ \hspace{-3mm} &     140 \hspace{-3mm}&  \hspace{-3mm}  5.1 \hspace{-3mm}&  \hspace{-1mm}$2.846942\cdot 10^{-6}$ \\
	\hspace{-1mm} $5 \cdot 10^6$ \hspace{-3mm} &  \hspace{-4.5mm}   750\hspace{-3mm} &   \hspace{-3mm} 30 \hspace{-3mm}&  \hspace{-2mm} $2.846941\cdot 10^{-6}$  \\
 \hline
 \end{tabular} 
 \caption{Compute times (in milliseconds) for an update of $\delta(\veps)$-bound for different values of $n$ using the summation~\eqref{eq:summ} and the inner product~\eqref{eq:planch0} and the $\delta(\veps)$-upper bound after $k=500$ compositions.
 We see that using Lemma~\ref{lem:plancherel} we can speed up the update more than 20-fold, and that accurate update of $\delta(\veps)$ is possible in less than one millisecond.
 }
 \label{table:plancherel}
\end{center}
\end{table}

\newpage 

\section{Experiments} \label{sec:experiments}

We compare experimentally the proposed method to the Tensorflow moments accountant~\citep{Abadi2016} which is based on RDP~\citep{mironov2017}
and allows evaluation of guarantees for heterogeneous compositions. 
For homogeneous compositions, in the Appendix we compare our method also to a more recent RDP accountant~\citep{zhu2019} 
and to Gaussian differential privacy (GDP) accounting~\citep{dong2021gaussian} as their existing implementations
 are not directly applicable to heterogeneous compositions.
In the Appendix we also illustrate the possible benefits obtained from using an improved conversion formula~\citep{asoodeh2020} from RDP to $(\veps,\delta)$-DP.

\subsection{Compositions of Discrete and Continuous Mechanisms}

We consider a non-adaptive composition of the form 
$$
\mathcal{M}(X) = \big(\mathcal{M}_1(X),\widetilde{\mathcal{M}}_2(X), \ldots, \mathcal{M}_{k-1}(X),\widetilde{\mathcal{M}}_k(X)\big),
$$
where each $\mathcal{M}_i$ is a Gaussian mechanism with sensitivity 1, and each $\widetilde{\mathcal{M}}_i$ is a randomised response mechanism with 
probability of a correct answer $p$, $\tfrac{1}{2} < p < 1$. We know that for the randomised response the PLD leading to the worst-case bound is given by
\begin{equation*}
	\begin{aligned}
\omega_{\mathrm{R}}(s) = p \cdot \deltas{c_p} + (1-p) \cdot \deltas{-c_p},
	\end{aligned}
\end{equation*}
where $c_p = \log \tfrac{p}{1-p}$~\citep{koskela2021tight}. Also, for the PLD $\omega_\mathrm{G}$ of the Gaussian mechanism we know that
$\omega_\mathrm{G} ~ \sim \mathcal{N}\left( \frac{1}{2 \sigma^2}, \frac{1}{\sigma^2} \right)$~\citep{sommer2019privacy}.
Let the $\Delta x$-grid be defined as above, i.e., 
let $L>0$, $n \in \mathbb{Z}^+$, $\Delta x = 2L/n$ and $s_i = -L + i \Delta x$ for all $i \in \mathbb{Z}$. 
Define
\begin{equation} \label{eq:c_plus}
	\begin{aligned}
		\omega_{\mathrm{G},\mathrm{max}}(s) = \sum\nolimits_{i=0}^{n-1} c^+_i \cdot \deltas{s_i}, \quad \textrm{where} \quad 
		c^+_i = \Delta x \cdot \max\nolimits_{s \in [s_{i-1}, s_i]} \omega_\mathrm{G}(s).
	\end{aligned}
\end{equation}
%
Using a bound for the moment generating function of the infinitely extending counterpart
of $\omega_{\mathrm{max}}$ and by using Alg.~\ref{alg:delta} (we refer to the Appendix for more details)
we obtain a numerical value $\delta_{\mathrm{max}}(\veps)$ (depending on $n$ and $L$) for which we have that $\delta(\veps) \leq \delta_{\mathrm{max}}(\veps)$,
where $\delta(\veps)$ gives a tight bound for the composition $\mathcal{M}(X)$.
As a comparison, in Figure~\ref{fig:rr} we also show the guarantees given by Tensorflow moments accountant.
We know that for $\alpha>1$, the $\alpha$-RDP of the randomised response is given by
$$
\frac{1}{\alpha - 1} \log \big( p^\alpha (1-p)^{1-\alpha} + (1-p)^\alpha  p^{1-\alpha}  \big)
$$
and correspondingly for the Gaussian mechanism by $\alpha/2\sigma^2$~\citep{mironov2017}. 
As is commonly done, we evaluate RDPs for integer values and sum up them along the compositions.
Then, using the moments accountant method the corresponding $(\veps,\delta)$-bounds are obtained~\citep{Abadi2016}.


\begin{figure}
     \centering
     \begin{subfigure}[b]{0.47\textwidth}
         \centering
         \includegraphics[width=\textwidth]{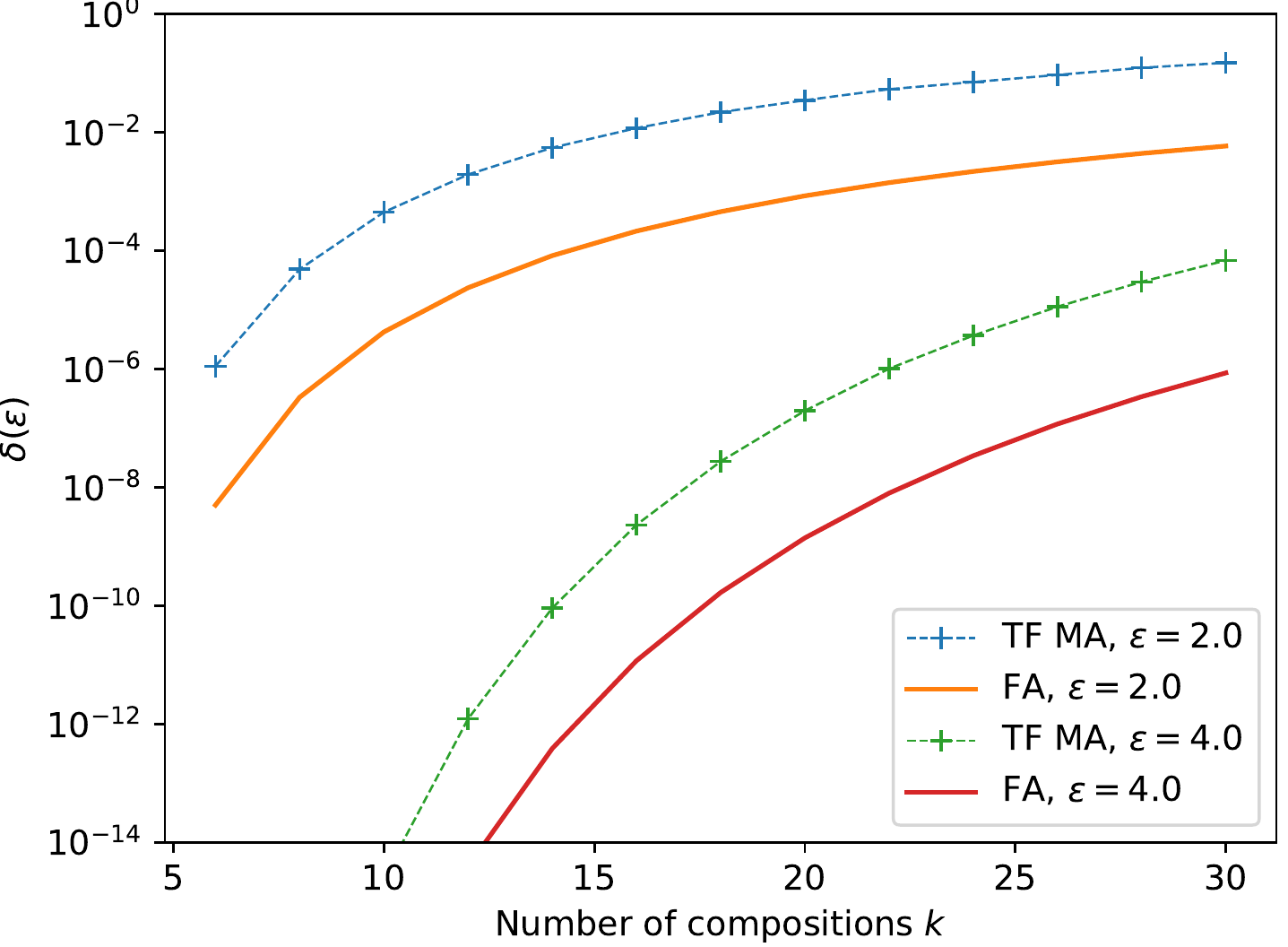}
         \caption{ Bounds for $\delta(\veps)$ computed using Algorithm~\ref{alg:delta} (FA) and Tensorflow moments accountant (TF MA),
	 when $\sigma=5.0$ and $p=0.52$, for $\veps=2.0,4.0$. 
	 We see that when $\delta \in [10^{-6},10^{-4}]$, FA allows approximately $1.5$ times as many compositions as TF MA for the same $\veps$. 
	 We use here $L=10$ and $n=10^5$ discretisation points, however note that already $n=5 \cdot 10^{-3}$ gives accurate results.}
		\label{fig:rr}
     \end{subfigure}
     \hfill
     \begin{subfigure}[b]{0.47\textwidth}
         \centering
         \includegraphics[width=\textwidth]{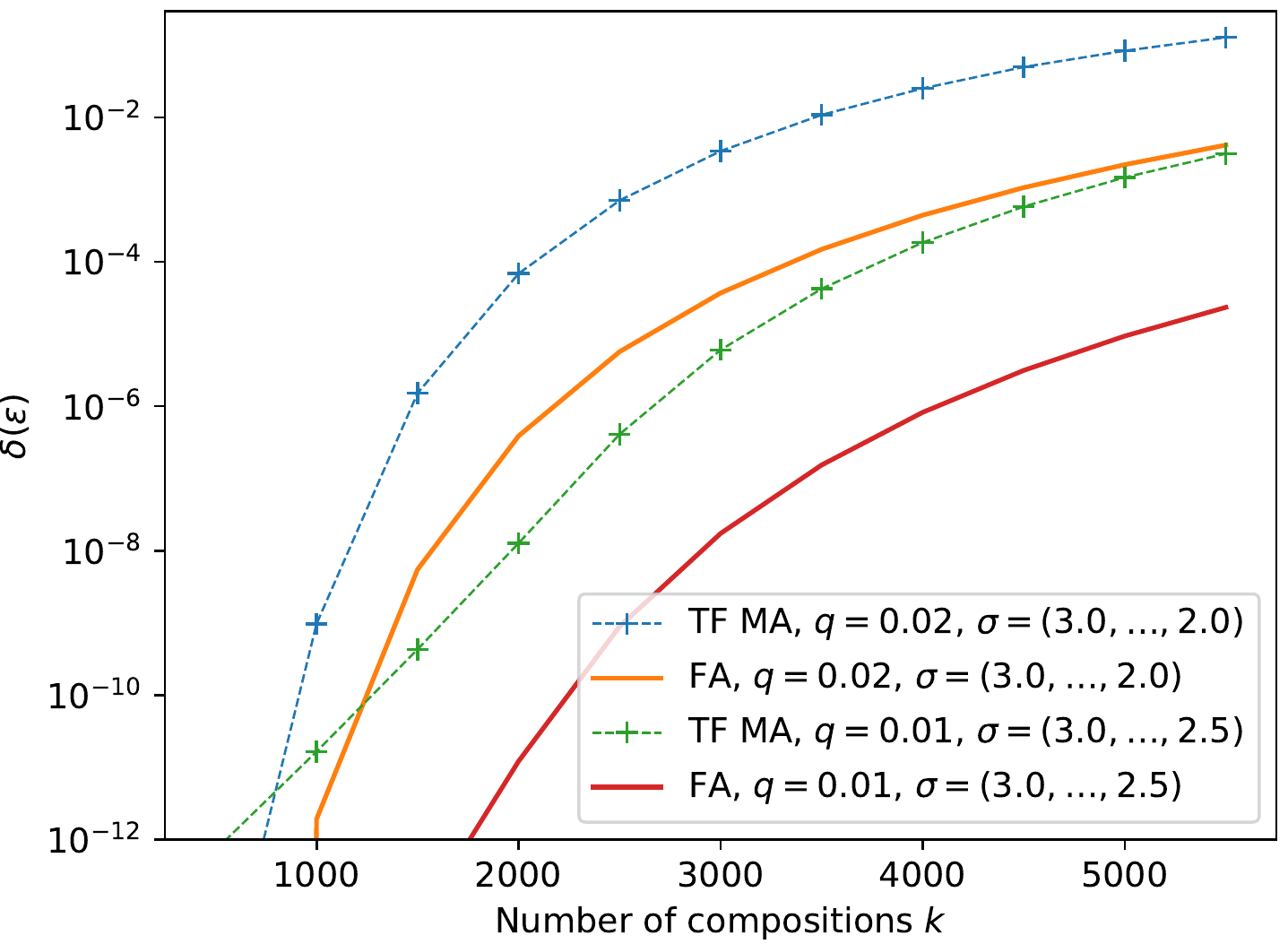}
         \caption{ Bounds for $\delta(\veps)$ computed using Algorithm~\ref{alg:delta} (FA) and Tensorflow moments accountant (TF MA).
	In the first option $\veps=1.0$, $q=0.02$ and $\sigma$ decreases linearly from $3.0$ to $2.0$.
 	In the second option $\veps=1.5$, $q=0.01$ and $\sigma$ decreases linearly from $3.0$ to $2.5$. For each value of $\sigma$, 500 compositions are evaluated.
	We see that when $\delta \in [10^{-6},10^{-4}]$, FA allows approx. $1.5$ times as many compositions. }
         \label{fig:subsampled}
     \end{subfigure}
	 \caption{Comparisons of FA and the Tensorflow moments accountant.}
\end{figure}
\subsection{Heterogeneous Subsampled Gaussian Mechanism} \label{subsec:subsampled}

We next show how to compute $(\veps,\delta)$-upper bounds for heterogeneous compositions of the subsampled Gaussian mechanism. 
We consider the Poisson subsampling and $\sim_R$-neighbouring relation.
The fact that we obtain an upper bound in this case by considering 
non-adaptive compositions of univariate mechanisms is shown in the Appendix.
For a subsampling ratio $q$ and noise level $\sigma$, the continuous PLD of the subsampled Gaussian mechanism is given by~\citep{koskela2020}
\begin{equation*} 
\omega(s) = \begin{cases}
f(g(s))g'(s), &\text{ if }  s > \log(1-q), \\
0, &\text{ otherwise},
\end{cases}
\end{equation*}
where
$$
f(t) = \frac{1}{\sqrt{2 \pi \sigma^2}} \, [ q \ee^{ \frac{-(t-1)^2}{2 \sigma^2}} + (1-q) \ee^{-\frac{t^2}{2 \sigma^2}} ], \quad
g(s) = \sigma^2 \log \left( \frac{\ee^s - (1-q)}{q} \right) + \frac{1}{2}.
$$
Analogously to \eqref{eq:c_plus}, using $\omega$ we determine a discrete PLD $\omega_{\mathrm{max}}$,
and by deriving a bound for the moment generating function 
of $\omega_{\mathrm{max}}$ (see also details in the Appendix) and by using Alg.~\ref{alg:delta} and Thm.~\ref{thm:alg_error_bound1}
we obtain a numerical value $\delta_{\mathrm{max}}(\veps)$  such that after $k$ compositions
\begin{equation*} 
	\delta(\veps) \leq \delta_{\mathrm{max}}(\veps),
\end{equation*} 
where $\delta(\veps)$ gives a tight bound for the $k$-fold  composition of heterogeneous subsampled Gaussian mechanisms.
Figure~\ref{fig:subsampled} illustrates $\delta_{\mathrm{max}}(\veps)$ as $k$ grows, when $L=10$ and $n=10^6$.
For comparison, we also show the numerical values given by Tensorflow moments accountant~\citep{Abadi2016}.

\section{Conclusions} 

We have extended the Fast Fourier Transform-based approach for computing tight privacy bounds for discrete-valued mechanisms to heterogeneous compositions.
We have given a complete error analysis of the method such that using the derived bounds it is possible to 
determine appropriate values for all the parameters of the algorithm, allowing more black-box like usage.
The error analysis also led to a complexity bound that is slightly better than the existing theoretical
complexity bound for non-adaptive compositions. Using the Plancherel theorem, we have shown how to further speed up the evaluation of DP bounds. 
We emphasise that due to the construction of the algorithm and to the rigorous error analysis, the reported $(\veps,\delta)$-bounds are strict upper privacy bounds.
One clear deficit of our approach, when compared to approaches such as GDP and RDP, is the difficulty of its implementation. 
However, in situations where accurate $(\veps,\delta)$-bounds for compositions of complex mechanisms are required,
the Fourier accountant appears as an attractive alternative.


\section*{Acknowledgements} 

This work has been supported by the Academy of Finland [Finnish Center for Artificial Intelligence FCAI and grant 325573] and by the Strategic Research Council at the Academy of Finland [grant 336032].

\newpage

\bibliography{pld}

\newpage

\appendix

\section{Comparisons to State-of-the-art DP Accountants}

%
%

\subsection{Comparison to the R\'enyi Differential Privacy accountant}

First, we compare our method to the RDP accountant by~\citet{zhu2019} which gives optimal RDP bounds for the subsampled Gaussian mechanism.
This method is included in the 'autodp' package~\footnote{https://github.com/yuxiangw/autodp} and it works for fixed values
of $\sigma$ and $q$. Computing Fourier accountant (FA) bounds in the case where $\sigma$ drops linearly from 3.0 to 2.5 and the subsampling ratio
$q$ is fixed, FA gives tighter bounds than the RDP accountant even for fixed $\sigma=3.0$ (see Fig.~\ref{fig:eps_rdp}).

\begin{figure} [h!]
	\centering
	\includegraphics[width=0.6\linewidth]{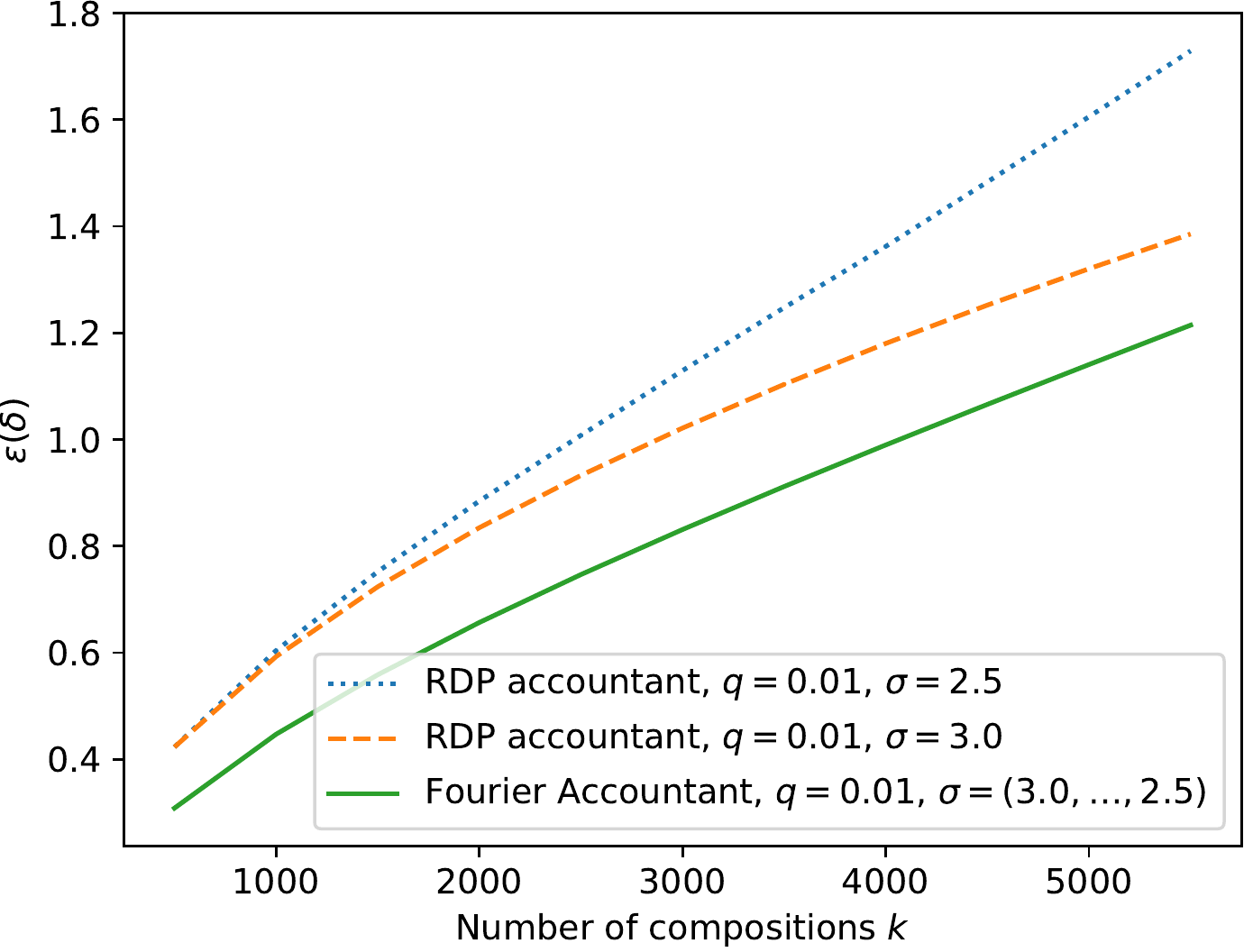}
\caption{Upper $\veps(\delta)$-bounds obtained using the Fourier Accountant and the R\'enyi DP accoutant 'autodp'. Here $\delta=10^{-6}$.}
\label{fig:eps_rdp}
\end{figure}

In Figures 2a and 2b we compare FA and 'autodp' in case $\sigma$ is fixed for both. We fix $\sigma=2.0$ and vary the subsampling ratio $q$
and the number of compositions $k$. We see that FA gives considerably tighter bounds.

Part of the differences in these results is explained by the loss in 
converting RDP-values to $(\veps,\delta)$-values.
The conversion of the RDP-values to $(\veps,\delta)$-values is carried out here using the formula~\citep{zhu2019} 
$
\delta(\veps)= \inf_{\alpha>1} \ee^{-(\alpha-1)(\veps-\gamma(\alpha)T)}.
$
In the next subsection consider the possible gains of using a tighter conversion formula.

\begin{figure} [h!]
     \centering
     \begin{subfigure}{0.48\textwidth}
         \centering
         \includegraphics[width=\textwidth]{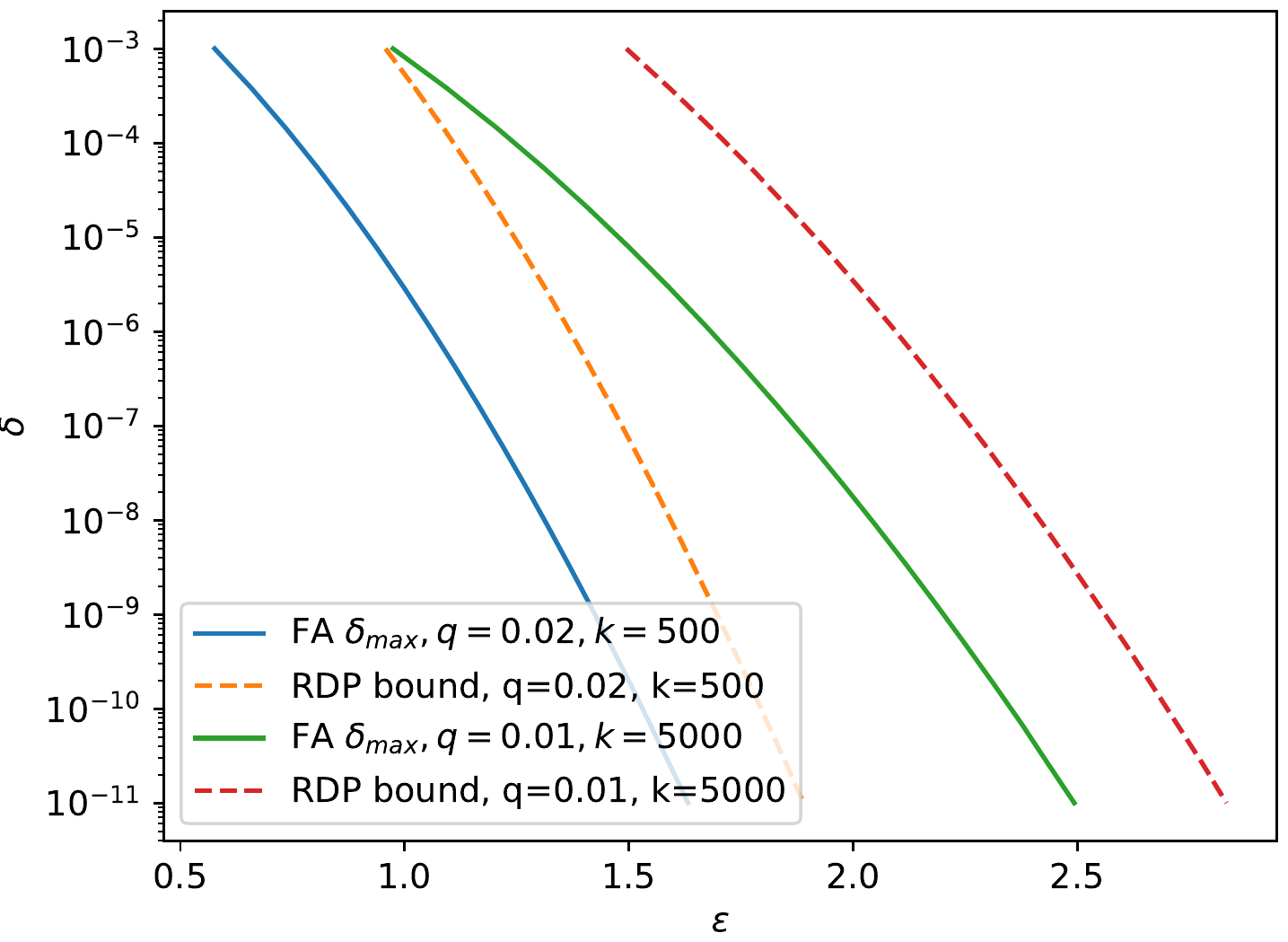}
         \caption{$q=0.01$, $k=5000$ and $q=0.02$, $k=500$.}
     \end{subfigure} 
     \begin{subfigure}{0.48\textwidth}
         \centering
         \includegraphics[width=\textwidth]{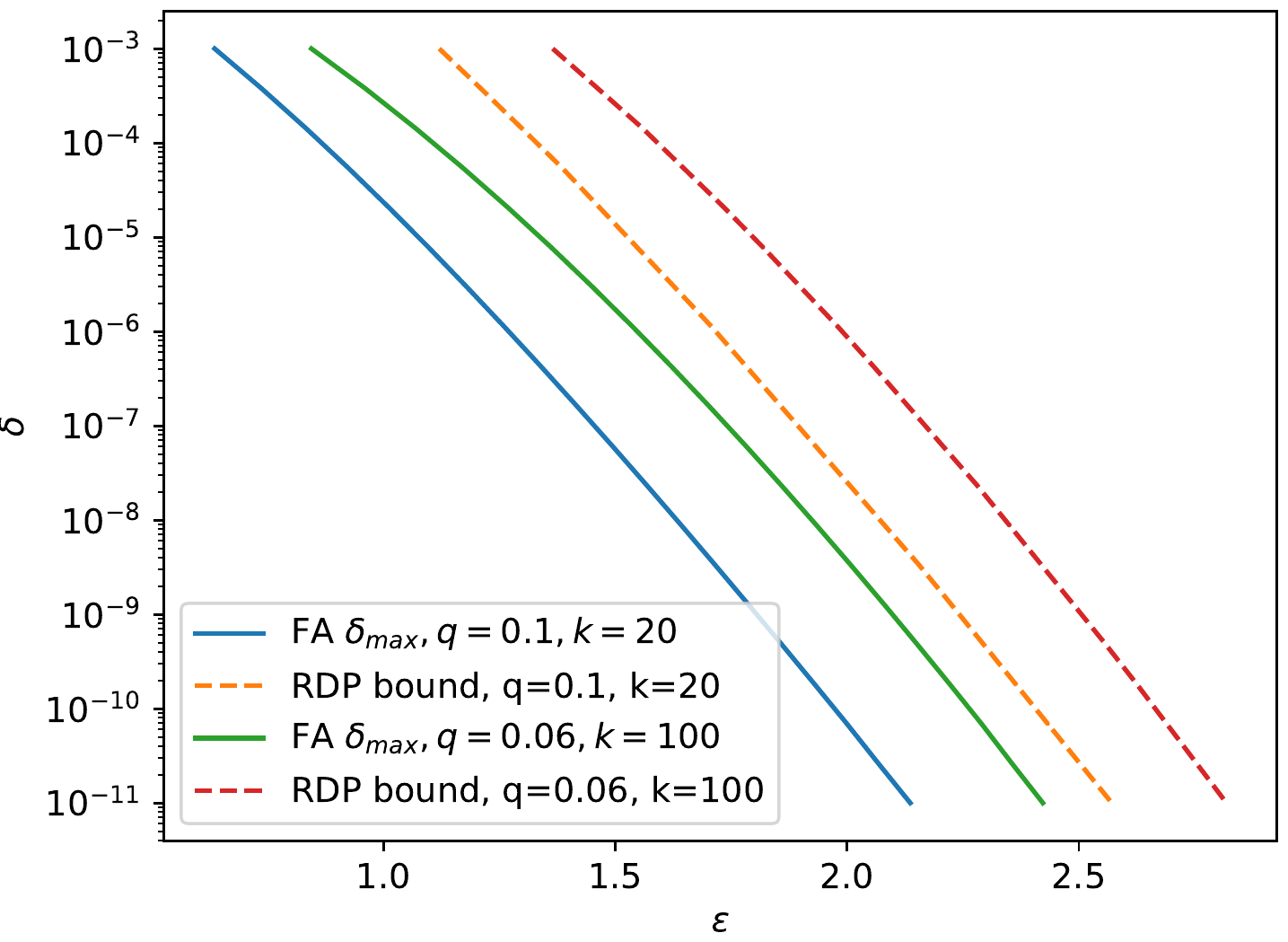}
         \caption{$q=0.06$, $k=100$ and $q=0.1$, $k=20$.}
     \end{subfigure} 
	 \caption{Comparison of the upper and lower $\delta(\veps)$-bounds given by FA and the RDP bound given by autodp,
	 for various configurations when $\sigma=2.0$. Here $\veps=1.0$.
	 } \label{fig:delta_rdp}
\end{figure}

\subsection{Tighter Conversion of RDP to $(\veps,\delta)$-DP}

\citet{asoodeh2020} consider a tighter conversion of RDP to $(\veps,\delta)$-DP. The bound is optimal in a sense, that
the obtained $(\veps,\delta)$-values satisfy
$$
\veps_\alpha^\delta(\gamma) = \inf \{ \veps \geq 0 \, : \, \forall \mathcal{M} \in \mathbb{M}_\alpha(\gamma) \textrm{ is } (\veps,\delta)-\textrm{DP} \},
$$
where $\mathbb{M}_\alpha(\gamma)$ denotes the set of all $(\alpha,\gamma)$-RDP mechanisms. 
As this definition suggests, the obtained $(\veps,\delta)$-bounds are not necessarily tight DP-bounds for a given particular mechanism.
Using log convex optimisation,~\citet{asoodeh2020} find $\veps_\alpha^\delta(\gamma)$-upper bounds
for the Gaussian mechanism from its RDP values~\citep[Lemma 2]{asoodeh2020}. We illustrate the sub-optimality of the resulting $(\veps,\delta)$-bounds as follows.

First of all, tight $(\veps,\delta)$-bounds for the Gaussian mechanism are obtained as follows. For the PLD $\omega_\mathrm{G}$ of the Gaussian mechanism we know that~\citep{sommer2019privacy}
$$
\omega_\mathrm{G} ~ \sim \mathcal{N}\left( \frac{1}{2 \sigma^2}, \frac{1}{\sigma^2} \right)
$$
and for a $k$-wise composition, by convolution, we have that 
$$
\omega^k_\mathrm{G} ~ \sim \mathcal{N}\left( \frac{k}{2 \sigma^2}, \frac{k}{\sigma^2} \right).
$$
The $(\veps,\delta)$-values for this PLD are obtained by conversion involving the CDF of the Gaussian function~\citep{sommer2019privacy}.

We know that the RDP-value of order $\alpha$ for the Gaussian mechanism is~\cite{mironov2017}
$$
\gamma(\alpha) = \alpha/2\sigma^2.
$$
We combine this RDP with the conversion formula of~\citep[Lemma 2]{asoodeh2020}.
We also compare the commonly used conversion formula~\citep[see e.g.][Thm.\;2]{Abadi2016}
\begin{equation} \label{eq:conversion_classic}
	\delta(\veps) = \inf_{\alpha>1} \ee^{-(\alpha-1)(\veps-\gamma(\alpha)T)}.
\end{equation}
As Fig.~\ref{fig:compare_DP_RDP} shows, the conversion by~\citet[Lemma 2]{asoodeh2020} gives tighter results than the commonly used conversion formula~\eqref{eq:conversion_classic},
however the $\veps_\alpha^\delta(\gamma)$-bound does not give tight   $(\veps,\delta)$-bounds whereas
the bounds given by the Fourier accountant converge to the tight $(\veps,\delta)$-bounds of the Gaussian mechanism.

\begin{figure}[h!]
	\centering
	\includegraphics[width=0.6\linewidth]{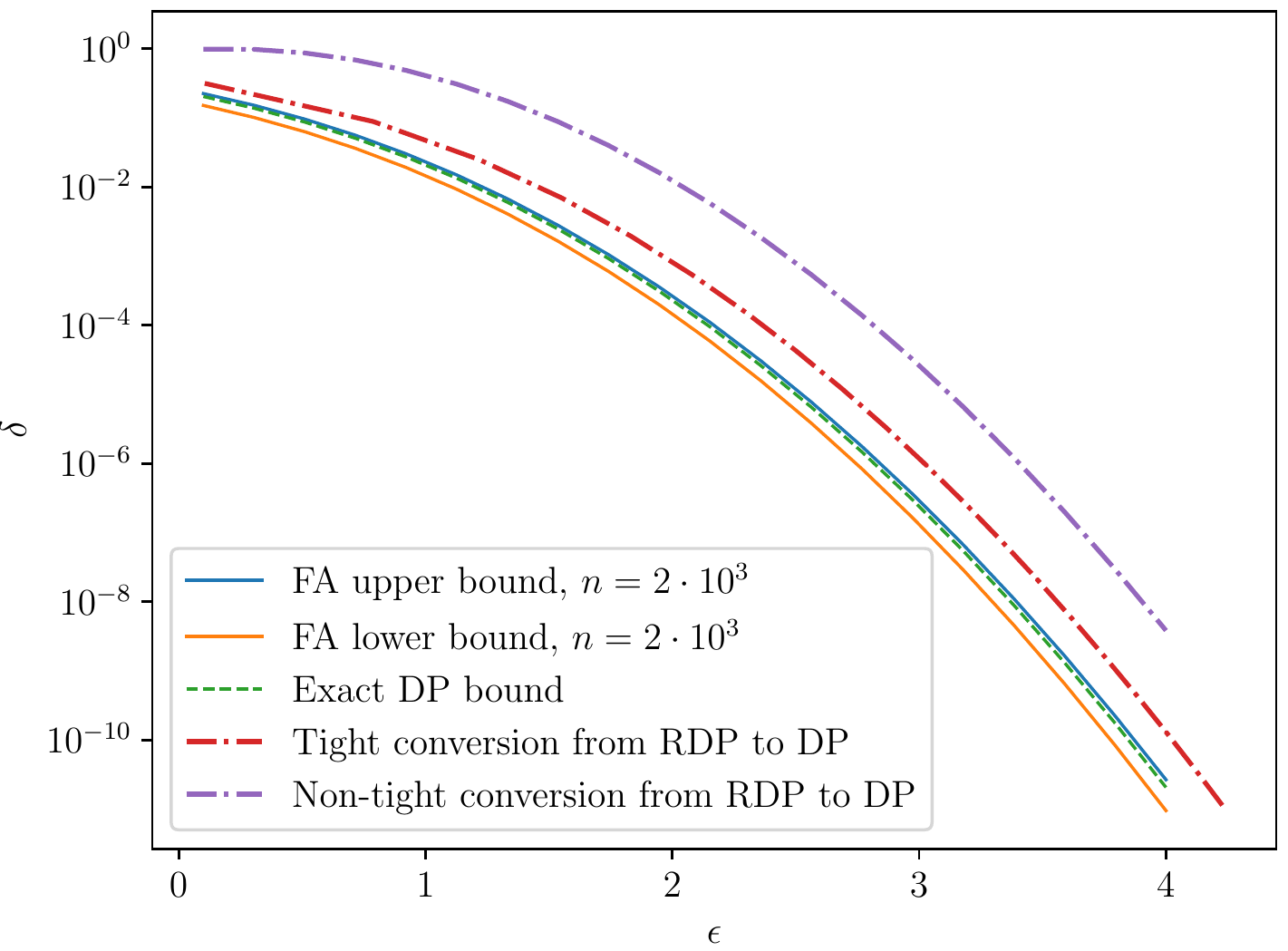}
\caption{Comparison of the Fourier Accountant and the RDP bounds obtained with different conversion methods.
Here $k=6$ compositions. Here $n$ denotes the number of discretisation points for FA. We note that already for $n=10^4$ the upper and lower bounds given by FA become almost indistinguishable.
}
\label{fig:compare_DP_RDP}
\end{figure}

\subsection{Comparison to the Gaussian Differential Privacy accountant}

Gaussian Differential Privacy is an attractive alternative for privacy accounting
as the bounds can be expressed using a single parameter $\mu$~\citep[for more details, see][]{dong2021gaussian}.
Conversion to $(\veps,\delta)$-bounds is straightforward using the CDF of the Gaussian function~\citep[Corollary 1]{dong2021gaussian}.
GDP gives exact $(\veps,\delta)$-bounds for compositions of the Gaussian mechanism. For other mechanisms, for large numbers of compositions one can approximate
the $\mu$-values using the central limit theorem. For example, in differentially private training of neural networks,
the number of compositions is commonly several tens of thousands which makes the resulting GDP approximates accurate.

~\citet[Section 4]{dong2021gaussian} provide also subsampling amplification results in case the subsample is of fixed size and uniformly sampled.
~\citet{bu2020deep} consider also the Poisson subsampling, 
and also an expression for the resulting DP bound is given, in terms of subsampling ratio $q$ and noise parameter 
$\sigma$~\citep[see Sec. 3][]{bu2020deep}. Evaluating this expression analytically is difficult and therefore~\citet{bu2020deep}
use the central limit theorem which says that after $k$ compositions the Poisson subsampled Gaussian mechanism
is approximately $p \sqrt{k(\ee^{1/\sigma^2}-1)}$-GDP. This formula combined with the conversion 
formula~\citep[Corollary 1][]{dong2021gaussian} is also the numerical method implemented in the Tensorflow libabry.

%

As these GDP results obtained using the Tensorflow accountant are \emph{approximations} based on the central limit theorem~\citep{bu2020deep}
instead of strict upper bounds like the results of the Fourier Accountant, we expect them to give inaccurate results for small numbers of compositions $k$.
This is indeed illustrated in Figure~\ref{fig:GDP_FA}. We emphasise that the first figure ($k=5000$) is closest to a realistic scenario of a DP-SGD training.

\begin{figure} [h!]
     \centering
     \begin{subfigure}{0.48\textwidth}
         \centering
         \includegraphics[width=\textwidth]{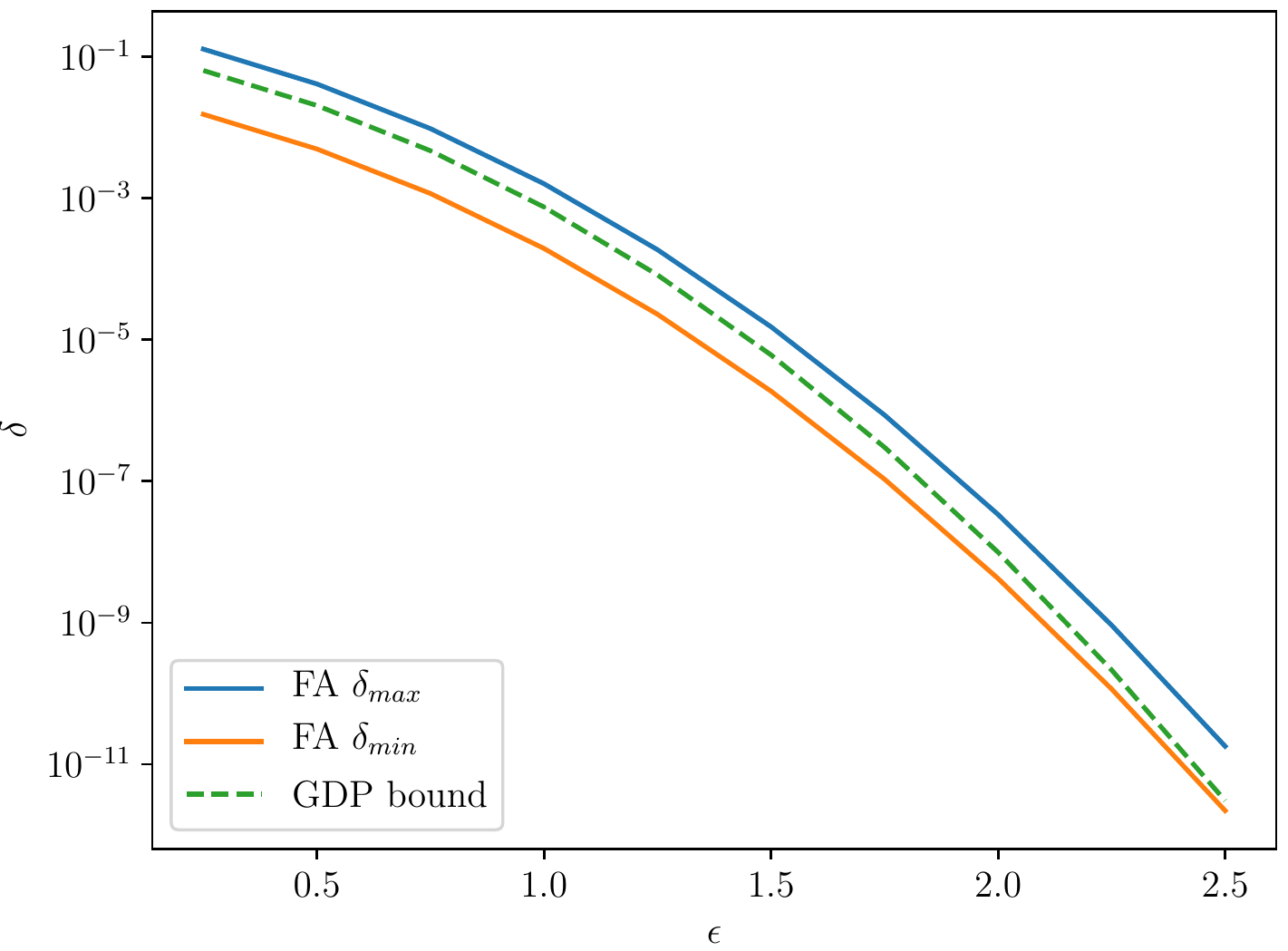}
         \caption{$n=10^7$, $q=0.01$, $k=5000$.}
     \end{subfigure}
     \hfill
     \begin{subfigure}{0.48\textwidth}
         \centering
         \includegraphics[width=\textwidth]{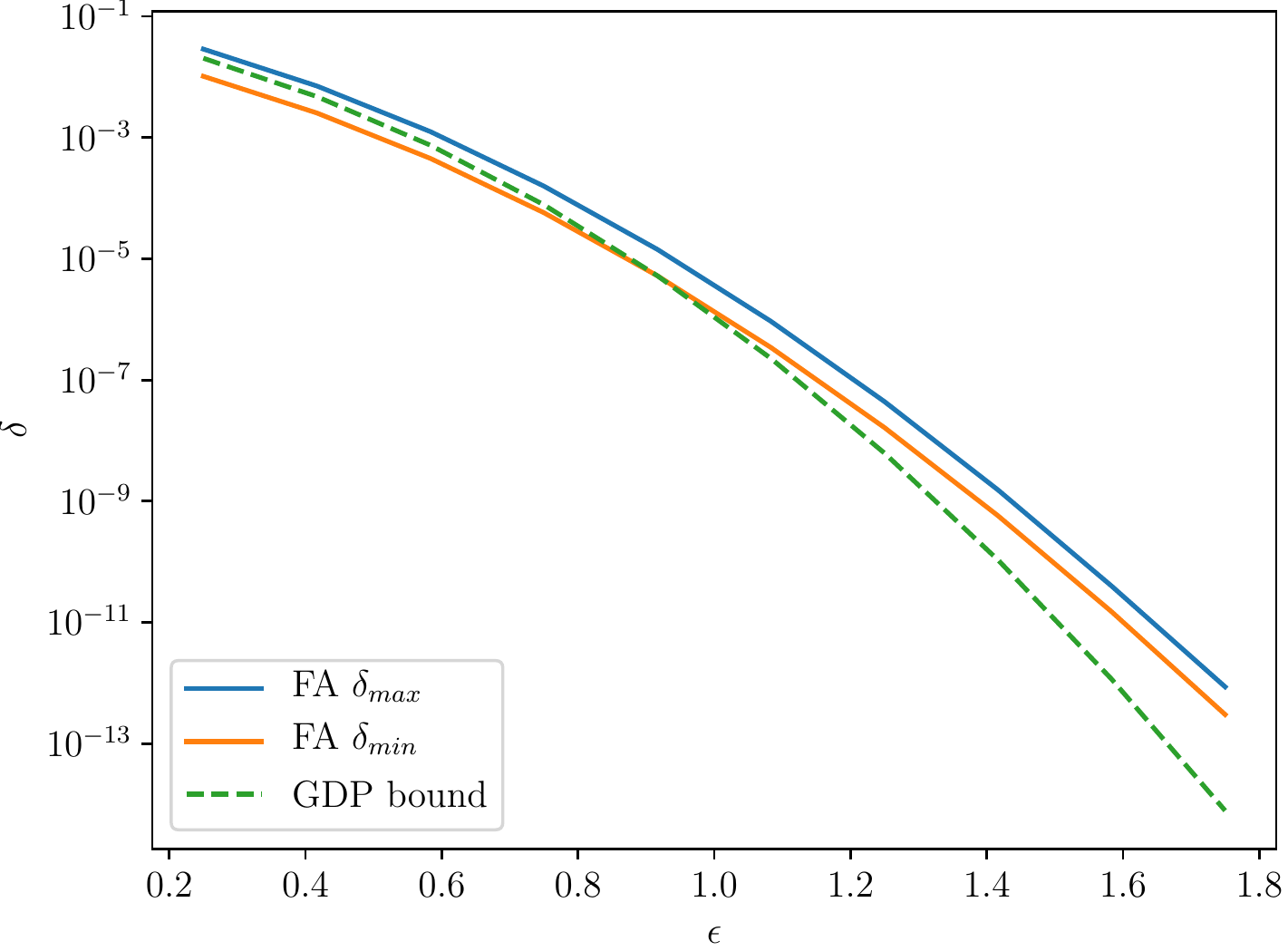}
         \caption{$n=2 \cdot 10^6$, $q=0.02$, $k=500$.}
     \end{subfigure} \\
     \begin{subfigure}{0.48\textwidth}
         \centering
         \includegraphics[width=\textwidth]{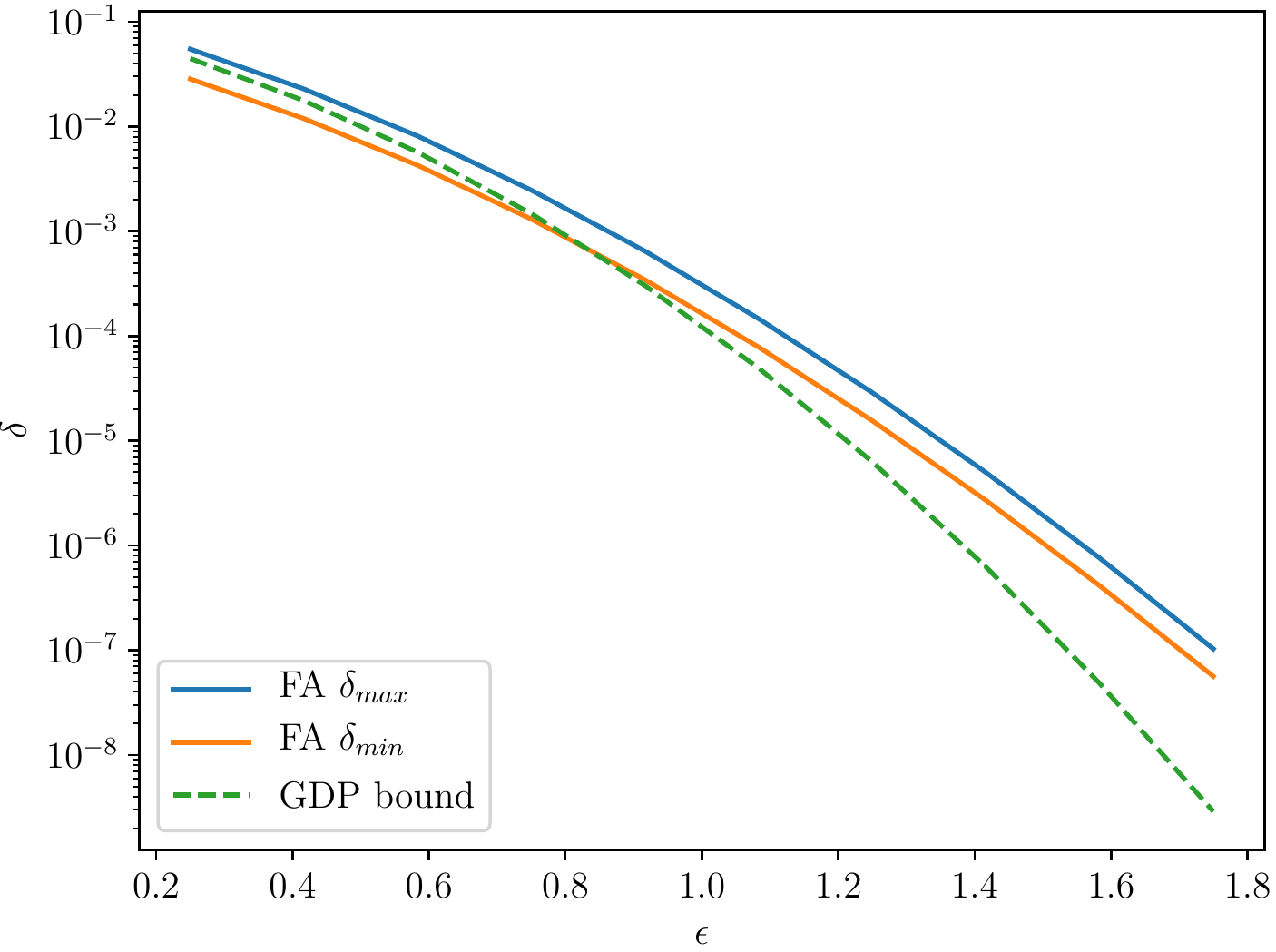}
         \caption{$n=2 \cdot 10^5$, $q=0.06$, $k=100$.}
     \end{subfigure}
     \hfill
     \begin{subfigure}{0.48\textwidth}
         \centering
         \includegraphics[width=\textwidth]{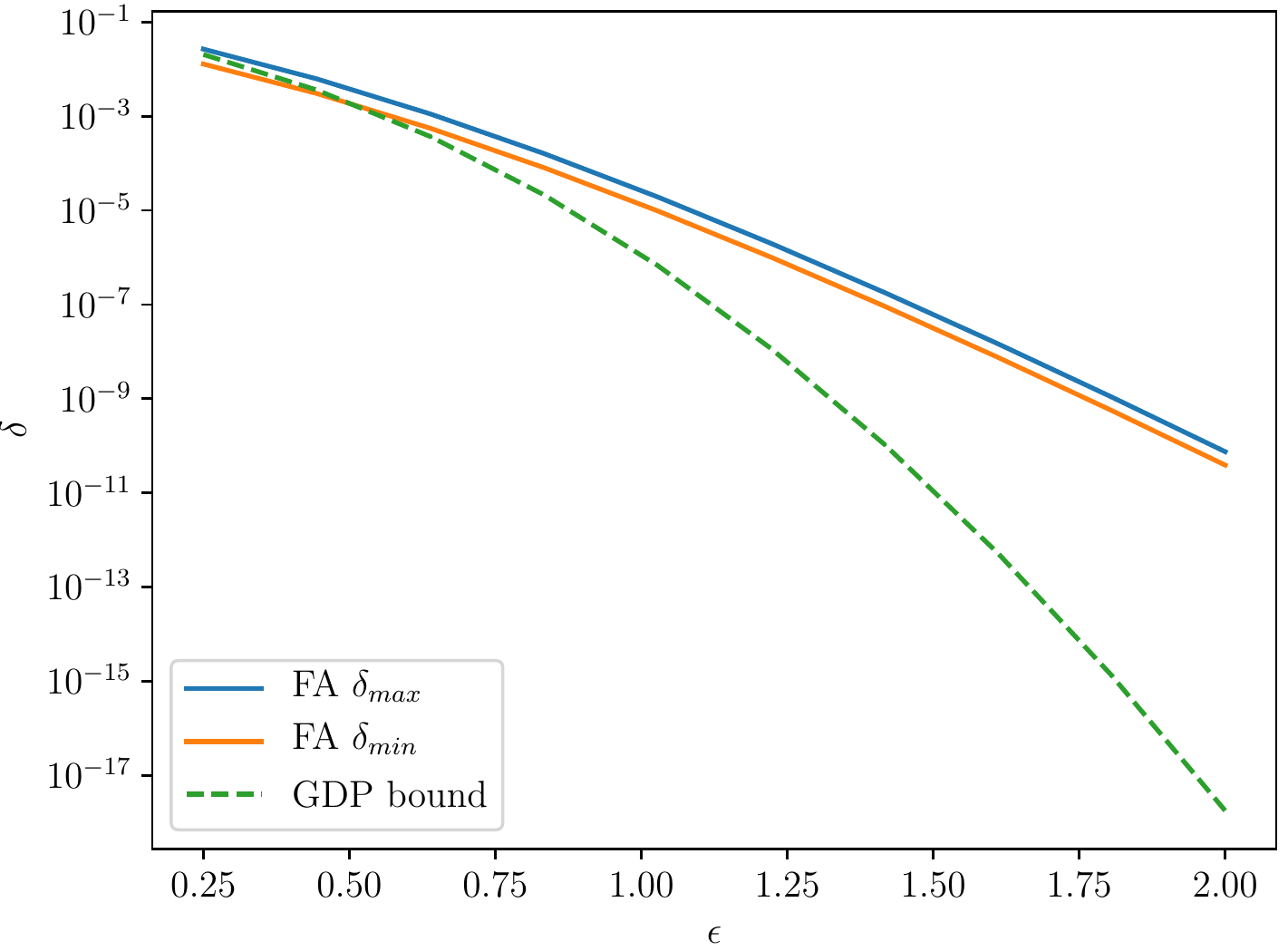}
         \caption{$n=2 \cdot 10^4$, $q=0.1$, $k=20$.}
     \end{subfigure} 
	 \caption{Comparison of the upper and lower $\delta(\veps)$-bounds given by the Fourier Accountant and the approximative GDP bound,
	 for different values of the subsampling ratio $q$ and different numbers of compositions $k$, when $\sigma=2.0$. 
	 For larger number of compositions, we use a larger number of discretisation points $n$ for the Fourier Accountant.
	 }	  \label{fig:GDP_FA}
\end{figure}

\subsection{Few Conclusions About the Comparisons}

Each of the DP accounting methods have their merits. Implementing GDP combined with a CLT approximation is extremely simple and for large numbers of homogeneous compositions
(i.e. compositions where the mechanisms do not vary) the approximations based on the CLT give accurate results, as shown by the experiments
of Figure~\ref{fig:GDP_FA}. However, for small number of compositions, the Fourier accountant appears as superior compared to this approach.

The situation is similar when using RDP: implementing the accountant and understanding its functionality is often easier than that of 
the Fourier accountant. 
The improved conversion bounds proposed by~\citet{asoodeh2020} seem to give considerably tighter DP bounds than the commonly used conversion formula~\citep{Abadi2016}.
However, as illustrated by Figures~\ref{fig:delta_rdp} and~\ref{fig:compare_DP_RDP}, the difference in the $\delta$-upper bounds obtained using RDP 
and the Fourier accountant remains approximately at a one order of magnitude. Moreover, as also shown by the experiments of~\cite{koskela2020}, for small number of
compositions the Fourier accountant appears to give upper $\delta(\veps)$-bounds that are several orders of magnitudes smaller.
One clear benefit of the RDP approach is that it is more easily applicable to heterogeneous \emph{adaptive compositions},
something that is more cumbersome when using the PLD approach.

The downside of the PLD approach is undoubtedly the complexity of the algorithm, there are simply much more lines of code involved and also possible pitfalls in the implementation.
However in situations where accurate $(\veps,\delta)$-bounds are required, and also for sanity-checking the functionality of other accountants
by using simple non-adaptive compositions, the Fourier Accountant appears as an attractive alternative.

\section{Theorem 4 of the Main Text}

Theorem 4 of the main text shows that the tight $(\veps,\delta)$-bounds for compositions 
of non-adaptive mechanisms are obtained using convolutions of PLDs (see also Thm.\;1 by~\citet{sommer2019privacy}).
 We include the proof here for completeness.

\begin{thm} \label{thm:convolutions} 
Consider a non-adaptive composition of $k$ independent mechanisms $\mathcal{M}_1,\ldots,\mathcal{M}_k$ and neighbouring data sets $X$ and $Y$.
The composition is tightly $(\veps,\delta)$-DP for $\delta(\veps)$ given by
$$
\delta(\veps) = \max \{ \delta_{X/Y}(\veps), \delta_{Y/X}(\veps) \},
$$ 
where
\begin{equation}  \label{eq:delta_inf}
	\begin{aligned}
		 \delta_{X/Y}(\veps) &= 1 - \prod\limits_{\ell=1}^k (1-\delta_{X/Y,\ell}(\infty)) 
		+ \int_\veps^\infty (1 - \ee^{\veps - s})\left(\omega_{X/Y,1} * \cdots * \omega_{X/Y,k} \right) (s)  \, \dd s, \\
 \delta_{X/Y,\ell}(\infty) &= 
\sum\limits_{ \{ t_i \, : \, \mathbb{P}( \mathcal{M_\ell}(X) = t_i) > 0, \, \mathbb{P}( \mathcal{M_\ell}(Y) = t_i) = 0 \} }
 		\mathbb{P}( \mathcal{M_\ell}(X) = t_i)
	\end{aligned}
\end{equation}
and $\omega_{X/Y,1} * \cdots * \omega_{X/Y,k}$ denotes the convolution of 
the density functions $\omega_{X/Y,\ell}$, $1 \leq \ell \leq k$. An analogous expression holds for $\delta_{Y/X}(\veps)$.
\begin{proof}
We show the proof first for the composition of two mechanisms. It will be clear from the proof how to generalise for a non-adaptive composition of $k$ mechanisms. 
We start by considering Lemma 4 of~\citep{koskela2021tight} that gives an expression for the tight $(\veps,\delta)$-DP bound for a single mechanism.
By definition of the privacy loss distribution, the PLD distribution $\widetilde{\omega}$ of the non-adaptive composition of mechanisms $\mathcal{M}_1$ and 
$\mathcal{M}_2$ is given by
\begin{equation*}
	\begin{aligned}
\widetilde{\omega}_{X/Y}(s) & = \sum\limits_{(t_i,t_i') = (t_j,t_j')}  \mathbb{P}\big( (\mathcal{M}_1(X),\mathcal{M}_2(X)  = (t_i,t_i') \big)
\cdot \deltas{ \widetilde{s}_i }, \\ 
& \quad \quad \widetilde{s}_i = \log \left( \frac{ (\mathcal{M}_1(X) , \mathcal{M}_2(X)) = (t_i,t_i') }
{ (\mathcal{M}_1(Y) , \mathcal{M}_2(Y) = (t_j,t_j') } \right).
	\end{aligned}
\end{equation*}
Due to the independence of $\mathcal{M}_1$ and $\mathcal{M}_2$,
\begin{equation} \label{eq:independence}
	\begin{aligned}
\mathbb{P}\big( \mathcal{M}_1(X) &=t_i, \, \mathcal{M}_2(X) = t_i'\big) = \mathbb{P}\big( \mathcal{M}_1(X)=t_i\big) \, \mathbb{P}\big( \mathcal{M}_2(X) = t_i'\big), \\
\mathbb{P}\big( \mathcal{M}_1(Y) &=t_j, \, \mathcal{M}_2(Y) = t_j' \big) = \mathbb{P}\big( \mathcal{M}_1(Y)=t_j\big) \, \mathbb{P}\big( \mathcal{M}_2(Y) = t_j'\big). \\
	\end{aligned}
\end{equation}
Therefore,
\begin{equation*}
	\begin{aligned}
		\log \left(  \frac{\mathbb{P}\big( \mathcal{M}_1(X)=t_i, \, \mathcal{M}_2(X) = t_i'\big)}{\mathbb{P}\big( \mathcal{M}_1(Y)=t_j, \, \mathcal{M}_2(Y) = t_j'\big) }  \right) 
		= \log \left( \frac{\mathbb{P}\big( \mathcal{M}_1(X)=t_i\big)}{ \mathbb{P}\big( \mathcal{M}_1(Y)=t_j\big) }    \right) + 
		\log \left( \frac{ \mathbb{P}\big( \mathcal{M}_2(X) = t_i'\big) }{ \mathbb{P}\big( \mathcal{M}_2(Y) = t_j'\big) }    \right).
	\end{aligned}
\end{equation*}
and
\begin{equation} \label{eq:conv_last_step}
	\begin{aligned}
\widetilde{\omega}_{X/Y}(s) = \sum\limits_{(t_i,t_i') = (t_j,t_j')}  \mathbb{P}\big( \mathcal{M}_1(X)=t_i\big) \, \mathbb{P}\big( \mathcal{M}_2(X) = t_i'\big)
\cdot \deltas{s_i + s_i'}, 
	\end{aligned}
\end{equation}
where
$$ 
s_i = \log \left( \frac{ \mathbb{P}\big( \mathcal{M}_1(X) = t_i \big) }{ \mathbb{P}\big( \mathcal{M}_1(Y) = t_j \big) } \right),
\quad s_i' = \log \left( \frac{ \mathbb{P}\big( \mathcal{M}_2(X) = t_i' \big) }{ \mathbb{P}\big( \mathcal{M}_2(Y) = t_j' \big) } \right).
$$
We see from \eqref{eq:conv_last_step} that $\widetilde{\omega}_{X/Y}  = \omega_{X/Y} * \omega_{X'/Y'}$ with discrete convolution $*$ as defined in the main text.
The expression for $\widetilde{\delta}_{X/Y}(\infty)$ follows directly from its definition in Lemma 4 of~\citep{koskela2021tight} 
that gives an expression for the tight $(\veps,\delta)$-DP bound for a single mechanism,
and from the independence of the mechanisms \eqref{eq:independence}. We see from this proof and from the definition of the discrete convolution 
that the result directly generalises for a non-adaptive composition of $k$ mechanisms.
\end{proof}
\end{thm}

\section{Proofs for the Results of Section 4}

\subsection{Lemma 5 of the Main Text}

We first prove Lemma 5 of the main text. To that end, recall the grid approximation: we place PLDs on a grid
\begin{equation} \label{Aeq:grid}
X_n = \{x_0,\ldots,x_{n-1}\}, \quad n \in \mathbb{Z}^+,
\end{equation}
where
$$
	x_i = -L + i \Delta x, \quad  \Delta x = 2L/n. 
$$
Suppose the PLD distribution $\omega$ is of the form 
\begin{equation} \label{Aeq:omega_0}
	\omega(s) = \sum\nolimits_{i=0}^{n-1} a_i \cdot \deltas{s_i},
\end{equation}
where $a_i \geq 0$ and $-L \leq s_i \leq L - \Delta x$, $0 \leq i \leq n-1$. 
We define the grid approximations
\begin{equation} \label{Aeq:omegaRL}
	\begin{aligned}
		\omega^\mathrm{L}(s)  &:= \sum\nolimits_{i=0}^{n-1} a_i \cdot \deltas{s_i^\mathrm{L}}, \\
		\omega^\mathrm{R}(s)  &:= \sum\nolimits_{i=0}^{n-1} a_i \cdot \deltas{s_i^\mathrm{R}}, 
	\end{aligned}
\end{equation}
where
\begin{equation*}
	\begin{aligned}
		   s_i^\mathrm{L} &= \max \{  x \in X_n \, : \, x \leq s_i  \}, \\ 
	       s_i^\mathrm{R} &= \min \{  x \in X_n \, : \, x \geq  s_i\}.
	\end{aligned}
\end{equation*}

\begin{lem} \label{Alem:deltaineq}
Let $\delta(\veps)$ be given by the integral formula of Theorem 4 of the main text for PLDs $\omega_1, \cdots, \omega_k$ of the form \eqref{Aeq:omega_0}.
Let $\delta^\mathrm{L}(\veps)$ and $\delta^\mathrm{R}(\veps)$  correspondingly be determined by the left and right 
approximations $\omega_1^\mathrm{L}, \ldots,  \omega_k^\mathrm{L} $ and $\omega_1^\mathrm{R}, \ldots,  \omega_k^\mathrm{R} $, as defined in \eqref{Aeq:omegaRL}.
Then for all $\veps>0$ :
\begin{equation} \label{Aeq:delta_ineq}
	\delta^\mathrm{L}(\veps) \leq \delta(\veps) \leq \delta^\mathrm{R}(\veps).
\end{equation}
\begin{proof}
Recall the integral formula of Theorem 4 of the main text:
$$
\delta(\veps) = 1 - \prod\limits_{\ell=1}^k (1-\delta_{\ell}(\infty)) + \int_\veps^\infty (1 - \ee^{\veps - s})\left(\omega_{1} * \cdots * \omega_{k} \right) (s)  \, \dd s.
$$
As the probabilities $\delta_{\ell}(\infty)$ are not affected by the grid approximation, we may only consider bounds for the integral
$$
\int_\veps^\infty (1 - \ee^{\veps - s})\left(\omega_{1} * \cdots * \omega_{k} \right) (s)  \, \dd s.
$$
By definition of the discrete convolution,
\begin{equation} \label{Aeq:cor_pld01}
	(\omega_1 * \cdots * \omega_k)(s) = \sum\nolimits_{i_1,\ldots,i_k=0}^{n-1} a_{i_1}^j \cdots a_{i_k}^j \cdot \deltas{s_{i_1} + \ldots + s_{i_k}}
\end{equation}
and 
\begin{equation} \label{Aeq:cor_pld02}
(\omega_1^\mathrm{L} * \cdots * \omega_k^\mathrm{L} )(s) = \sum\nolimits_{i_1,\ldots,i_k=0}^{n-1} a_{i_1}^j \cdots a_{i_k}^j \cdot \deltas{s_{i_1}^\mathrm{L} + \ldots + s_{i_k}^\mathrm{L}}.
\end{equation}
Since $(1-\ee^{\veps - s})$ is a monotonously increasing function of $s$ for $s \geq \veps$, and 
since $s_{i_1} + \ldots + s_{i_k} \geq s_{i_1}^\mathrm{L} + \ldots + s_{i_k}^\mathrm{L}$ for all $(i_1,\ldots,i_k)$, we instantly see from \eqref{Aeq:cor_pld01} and \eqref{Aeq:cor_pld02} that
$$
\delta^\mathrm{L}(\veps) \leq \delta(\veps).
$$
For the right grid approximation, $s_{i_1} + \ldots + s_{i_k} \leq s_{i_1}^\mathrm{R} + \ldots + s_{i_k}^\mathrm{R}$ for all $(i_1,\ldots,i_k)$, and we similarly see that 
$$
\delta^\mathrm{R}(\veps) \geq \delta(\veps).
$$
\end{proof}
\end{lem}
\subsection{Lemma 6 of the Main Text}

We next prove Lemma 6 of the main text which shows that the truncated convolutions of periodic distributions
can be evaluated using FFT. 
Suppose $\omega_1$ and $\omega_2$ are defined such that 
\begin{equation} \label{Aeq:omega}
	\omega_1(s) = \sum\nolimits_i a_i \cdot \deltas{s_i}, \quad \omega_2(s) = \sum\nolimits_i b_i \cdot \deltas{s_i},
\end{equation}
where for all $i$: $a_i,b_i \geq 0$ and $s_i = i \Delta x$. 
The convolution $\omega_1 * \omega_2$ can then be written as 
\begin{equation*}
	\begin{aligned}
	(\omega_1 * \omega_2)(s)  & =  \sum\nolimits_{i,j} a_i b_j \cdot \deltas{s_i + s_j} \\
	& = \sum\nolimits_i \Big(\sum\nolimits_j a_j b_{i-j} \Big) \cdot \deltas{s_i}.
\end{aligned}
\end{equation*}
Let $L>0$. We truncate these convolutions to the interval $[-L,L]$ such that 
\begin{equation*}
	\begin{aligned}
	(\omega_1 * \omega_2 )(s) & \approx \sum\nolimits_i \Big(\sum\nolimits_{-L \leq s_j < L} a_j b_{i-j} \Big) \cdot \deltas{s_i}   \\
	& =: (\omega_1 \circledast \omega_2 )(s).
\end{aligned}
\end{equation*}
For $\omega_1$ of the form \eqref{Aeq:omega}, we define $\widetilde{\omega}_1$ to be a $2 L$-periodic extension of $\omega_1$
from $[-L,L]$ to $\mathbb{R}$, i.e., $\widetilde{\omega}_1$ is of the form
$$
\widetilde{\omega}_1(s) = \sum\nolimits_{m \in \mathbb{Z}} \, \sum\nolimits_i a_i \cdot \deltas{s_i + m \cdot 2 L}.
$$
For $\omega_1$ and $\omega_2$ of the form \eqref{Aeq:omega}, we approximate the convolution $\omega_1 * \omega_2$ as
\begin{equation*}
	\omega_1 * \omega_2 \approx \widetilde{\omega}_1 \circledast \widetilde{\omega}_2.
\end{equation*}
Since $\omega_1$ and $\omega_2$ are defined on an equidistant grid, FFT can be used to evaluate the approximation $\widetilde{\omega}_1 \circledast \widetilde{\omega}_2$
as follows:
\begin{lem}[Lemma 6 of the main text] \label{Alem:fft}
Let $\omega_1$ and $\omega_2$ be of the form \eqref{Aeq:omega}, such that $s_i = -L + i \Delta x$, $0 \leq i \leq n-1$, where $L>0$, $n$ is even and 
$\Delta x = 2L/n$.
Define
\begin{equation*} 
 	\begin{aligned}
\boldsymbol{a} & = \begin{bmatrix} a_0 & \ldots & a_{n-1} \end{bmatrix}^\mathrm{T}, \\
\quad  \boldsymbol{b} &= \begin{bmatrix} b_0 & \ldots & b_{n-1} \end{bmatrix}^\mathrm{T}, \\
 D &= \begin{bsmallmatrix} 0 & I_{n/2} \\ I_{n/2} & 0 \end{bsmallmatrix} \in \mathbb{R}^{n \times n}.
	\end{aligned}
\end{equation*} 
Then, 
$$
(\widetilde{\omega}_1 \circledast \widetilde{\omega}_2 )(s) = \sum\nolimits_{i=0}^{n-1} c_i \cdot \deltas{s_i},
$$
where
$$
c_i = \left[D \, \mathcal{F}^{-1} \big(\mathcal{F}( D \, \boldsymbol{a} ) \odot \mathcal{F}( D \, \boldsymbol{b} )    \big) \right]_i,
$$
and $\odot$ denotes the element-wise product of vectors.
\end{lem}
\begin{proof}
Assume $n$ is even and $s_i = -L + i \Delta x$, $0 \leq i \leq n-1$. 
From the the truncation and periodisation it follows that
$\widetilde{\omega} \circledast \widetilde{\omega}$ is of the form
\begin{equation} \label{Aeq:b_i}
	(\widetilde{\omega}_1 \circledast \widetilde{\omega}_2 )(s) = \sum\limits_{i=0}^{n-1} c_i \cdot \deltas{s_i},
	\quad \quad c_i = \sum\limits_{j = n/2}^{3n/2-1} a_j \, b_{i - j } \, \textrm{ (indices modulo $n$)}.
\end{equation}
Denoting $\boldsymbol{\widetilde{a}} = D \boldsymbol{a}$ and $\boldsymbol{\widetilde{b}} = D \boldsymbol{b}$,
we see that the coefficients $c_i$ in \eqref{Aeq:b_i} are given by the expression
$$
c_{i+n/2} = \sum\limits_{j = 0}^{n-1} \widetilde{a}_j \, \widetilde{b}_{i - j } \, \textrm{ (indices modulo $n$)},
$$
to which we can apply DFT and the convolution theorem \citep{stockham1966}. Thus, when $0 \leq i \leq n-1$,
\begin{equation} \label{Aeq:FinvF}
	c_{i+n/2} = \left[ \mathcal{F}^{-1} \big(\mathcal{F}( \boldsymbol{\widetilde{a}} ) \odot  \mathcal{F} ( \boldsymbol{\widetilde{b}} )  \big) \right]_i 
	= \left[ \mathcal{F}^{-1} \big(\mathcal{F}( D \boldsymbol{a} ) \odot  \mathcal{F} ( D \boldsymbol{b} )  \big) \right]_i,  \, \textrm{ (indices modulo $n$)}
\end{equation}
where $\odot$ denotes the elementwise product of vectors. From \eqref{Aeq:FinvF} we find that
$$
	c_i = \left[ D \mathcal{F}^{-1} \big(\mathcal{F}( D \boldsymbol{a} ) \odot  \mathcal{F} ( D \boldsymbol{b} )  \big) \right]_i,  \, \textrm{ (indices modulo $n$)}.
$$
\end{proof}
%

\section{Proofs for the Results of Section 5} \label{Asec:err_est_a} 


\subsection{Decomposition of the Total Error}

%
When carrying out the approximations described in the main text, we 
\begin{enumerate}
	\item First replace the PLDs $\omega_1,\ldots,\omega_k$ by the right grid approximations
$\omega_1^\mathrm{R},\ldots,\omega_k^\mathrm{R}$. Using the notation given in the main text, this corresponds to the approximation $\delta(\veps) \approx \delta^\mathrm{R}(\veps)$,
i.e. to the approximation
$$
\int\limits_L^\infty (1 - \ee^{\veps - s})(\omega_1 * \cdots * \omega_k ) (s)  \, \dd s \approx 
\int\limits_L^\infty (1 - \ee^{\veps - s})(\omega_1^\mathrm{R} * \cdots * \omega_k^\mathrm{R} ) (s)  \, \dd s.
$$
	\item Then, the Fourier accountant is used to approximate $\delta^\mathrm{R}(\veps) \approx \widetilde{\delta^\mathrm{R}}(\veps)$ which  in exact arithmetic corresponds to the approximation
$$
\int\limits_L^\infty (1 - \ee^{\veps - s})(\omega_1^\mathrm{R} * \cdots * \omega_k^\mathrm{R} ) (s)  \, \dd s \approx 
\int\limits_L^\infty (1 - \ee^{\veps - s})(\widetilde{\omega}_1^\mathrm{R} \circledast \cdots \circledast \widetilde{\omega}_k^\mathrm{R} ) (s)  \, \dd s,
$$
where $\widetilde{\omega}_i$'s denote the periodised PLD distributions and $\circledast$ denotes the truncated convolutions (described in the main text).
\end{enumerate}
We separately consider the errors arising from the periodisation and truncation of the convolutions and from the grid approximation. 
This means that we bound the total error as
\begin{equation*}
	\begin{aligned}
		\abs{ \delta(\veps) - \widetilde{\delta^\mathrm{R}}(\veps) } & = \abs{ \delta(\veps) - \delta^\mathrm{R}(\veps) + \delta^\mathrm{R}(\veps) - \widetilde{\delta^\mathrm{R}}(\veps)} \\
		& \leq \abs{ \delta(\veps) - \delta^\mathrm{R}(\veps)} + \abs{\delta^\mathrm{R}(\veps) - \widetilde{\delta^\mathrm{R}}(\veps)}
	\end{aligned}
\end{equation*}
Theorem 8 of the main text gives a bound for the term $\abs{ \delta(\veps) - \delta^\mathrm{R}(\veps)}$ and
Theorem 7 bounds for the term
$\abs{\delta(\veps) - \widetilde{\delta}(\veps)}$, in terms of the moment generating functions (MGFs) of $\omega_1,\ldots,\omega_k$ and $-\omega_1,\ldots,-\omega_k$. 
The bounds for the error $\abs{\delta(\veps) - \widetilde{\delta}(\veps)}$ can be directly used to bound 
the error $\abs{\delta^\mathrm{R}(\veps) - \widetilde{\delta^\mathrm{R}}(\veps)}$, either by numerically evaluating the MGFs of the PLDs 
$\omega_1^\mathrm{R},\ldots,\omega_k^\mathrm{R}$, or by using MGFs of the PLDs $\omega_1,\ldots,\omega_k$ and Lemma 7 of~\cite{koskela2021tight}, which states that when $0 < \lambda < (\Delta x)^{-1}$,
\begin{equation*} 
	\mathbb{E} [ \ee^{ - \lambda \omega^\mathrm{R}} ] \leq \mathbb{E} [ \ee^{ - \lambda \omega } ] \quad \textrm{and} \quad
	\mathbb{E} [ \ee^{ \lambda \omega^\mathrm{R}} ] \leq \tfrac{1}{1-\lambda \Delta x} \mathbb{E} [ \ee^{ \lambda \omega } ],
\end{equation*}
where $\Delta x = 2L/n$.

\subsection{Tail Bound for the Convolved PLDs}

For the error analysis we repeatedly use the Chernoff bound~\citep{wainwright2019}
\begin{equation*} 
\mathbb{P}[ X \geq t] = \mathbb{P}[ \ee^{\lambda X} \geq \ee^{\lambda t} ] \leq \frac{ \mathbb{E}[ \ee^{\lambda X} ] }{\ee^{\lambda t}}
\end{equation*}
which holds for any random variable $X$ and for all $\lambda > 0$.
If $\omega$ is of the form 
$$
\omega(s) = \sum_{i=0}^{n-1} a_i \cdot \deltas{s_i}, \quad s_i = \log \left( \frac{a_{X,i}}{a_{Y,i}}   \right),
$$ 
where $a_{X,i},a_{Y,i} \geq 0$, $s_i \in \mathbb{R}$, $0 \leq i \leq n-1$, the moment generating function is given by
\begin{equation} \label{Aeq:pld_lmf}
	\begin{aligned}
		\mathbb{E} [\ee^{\lambda \omega_{X/Y} }] &= \int\limits_{-\infty}^\infty \ee^{\lambda s} \omega(s) \, \dd s 
		= \sum\limits_{i=1}^n  \ee^{\lambda s_i} \cdot a_{X,i} 
		= \sum\limits_{i=1}^n  \left( \frac{a_{X,i}}{a_{Y,i}}  \right)^\lambda a_{X,i}.
	\end{aligned}
\end{equation}

Denote $S_k := \sum_{i=1}^k \omega_i$, where $\omega_i$ denotes the PLD random variable of the $i$th mechanism.
Since $\omega_i$'s are independent, we have that
$$
\mathbb{E} [ \ee^{\lambda S_k}  ]  = \prod\nolimits_{i=1}^k \mathbb{E} [ \ee^{\lambda \omega_i}  ].
$$
Then, the Chernoff bound shows that for any $\lambda > 0$
\begin{equation} \label{Aeq:chernoff}
	\begin{aligned}
	I_1(L) &= 	\int_L^\infty  ( \omega_{1} * \cdots * \omega_{k})(s) \, \dd s \\
	&= \mathbb{P}[ S_k \geq L ] \\
		&\leq \prod\nolimits_{i=1}^k \mathbb{E} [ \ee^{\lambda \omega_i}  ] \,  \ee^{- \lambda L} \\
		& \leq \ee^{\sum_{i=1}^k \alpha_i(\lambda)}  \ee^{- \lambda L},
	\end{aligned}
\end{equation}
where $\alpha_i(\lambda) = \log( \mathbb{E} [\ee^{\lambda \omega_i }] )$.

\subsection{Theorem 7 of the Main Text} \label{Asubsec:third_approx}

We next give a proof for Theorem 7 of the main text.
Recall: denote the logarithms of the moment generating functions of the PLDs as
$$
\alpha_i^+(\lambda) = \log \Big(	\mathbb{E} [\ee^{\lambda \omega_i }] \Big),
\quad  \alpha_i^-(\lambda) = \log \Big(	\mathbb{E} [\ee^{ - \lambda \omega_i }] \Big),
$$
where $1 \leq i \leq k$. Futhermore, denote 
\begin{equation*}
	\alpha^+(\lambda) = \sum\nolimits_i \, \alpha^+_i(\lambda),
	\quad  \alpha^-(\lambda) = \sum\nolimits_i \, \alpha^-_i(\lambda).
\end{equation*}
Using the Chernoff bound, we obtain the required using $\alpha^+(\lambda)$ and $\alpha^-(\lambda)$.

\begin{remark}
Notice that in case $s_j \in [-L,L]$ for all $j$, $0 \leq j \leq n-1$, then 
$$
(\omega_1 \circledast \cdots \circledast \omega_k)(s) = (\omega_1 * \cdots * \omega_k)(s),
$$
i.e., the error arising from the truncation of discrete convolutions vanishes and Thm. 7 of the main text
gives the total error arising from periodisation and truncation operations when $s_j \in [-L,L]$ for all $j$.

\end{remark}

\begin{thm}[Thm. 7 of the main text] \label{Alem:period}
Let $\omega$ be defined as above and suppose $s_i \in [-L,L]$ for all $0\leq i \leq n-1$. Then, 
\begin{equation*}
	\begin{aligned}
 I_3(L) = \int\limits_\veps^L\abs{(\omega_1 \circledast \cdots \circledast \omega_k - \widetilde{\omega}_1 \circledast \cdots  \circledast\widetilde{\omega}_k)(s)} \, \dd s  
 \leq \big(  \ee^{\alpha^+(\lambda)} + \ee^{\alpha^-(\lambda)}  \big) \frac{ \ee^{- L \lambda} }{ 1 - \ee^{- 2L \lambda}}.
	\end{aligned}
\end{equation*}	
\begin{proof}

Let $\omega_i$'s and the corresponding $2L$-periodic continuations $\widetilde{\omega}_i(s)$ be of the form
$$
\omega_i(s) = \sum_i a_j^i \cdot \deltas{s_j} \quad \textrm{and} \quad 
\widetilde{\omega}_i(s) = \sum_j \widetilde{a}_j^i \cdot \deltas{s_j},
$$
where $s_j = j \Delta x$ and $a_j^i,\widetilde{a}_j^i \geq 0$. By definition of the truncated convolution $\circledast$, 
\begin{equation*}
	\begin{aligned}
	( \widetilde{\omega}_1 \circledast \cdots  \circledast\widetilde{\omega}_k)(s) & = 
	\sum\limits_{-L \leq s_{j_1} < L} \widetilde{a}_{j_1}^1 \sum\limits_{-L \leq s_{j_2} < L} \widetilde{a}_{j_2}^2 
	\ldots \sum\limits_{-L \leq s_{j_{k-1}} < L} \widetilde{a}_{j_{k-1}}^{k-1} 
	\sum\limits_i  \widetilde{a}_{i - j_1 - \ldots - j_{k-1}}^k \cdot \deltas{s_i}  \\
	& = \sum\limits_{-L \leq s_{j_1} < L} a_{j_1}^1 \sum\limits_{-L \leq s_{j_2} < L} a_{j_2}^2 
		\ldots \sum\limits_{-L \leq s_{j_{k-1}} < L} a_{j_{k-1}}^{k-1} 
		\sum\limits_i  \widetilde{a}_{i - j_1 - \ldots - j_{k-1}}^k \cdot \deltas{s_i} \\
	& = \sum\limits_{j_1} a_{j_1}^1 \sum\limits_{j_2}  a_{j_2}^2 
		\ldots \sum\limits_{j_{k-1}}  a_{j_{k-1}}^{k-1}  
		\sum\limits_i  \widetilde{a}_{i - j_1 - \ldots - j_{k-1}}^k \cdot \deltas{s_i},
	\end{aligned}
\end{equation*}	
since for all $i$, $\widetilde{a}^i_j = a^i_j$ for all $j$ such that $-L \leq s_j < L$.
Furthermore,
\begin{equation*}
	\begin{aligned}
	(\omega_1 \circledast \cdots \circledast \omega_k)(s)
	& = \sum\limits_{-L \leq s_{j_1} < L} a_{j_1}^1 \sum\limits_{-L \leq s_{j_2} < L} a_{j_2}^2
		\ldots \sum\limits_{-L \leq s_{j_{k-1}} < L} a_{j_{k-1}}^{k-1} 
		\sum\limits_i  a_{i - j_1 - \ldots - j_{k-1}}^k \cdot \deltas{s_i} \\
	& = \sum\limits_{j_1} a_{j_1}^1 \sum\limits_{j_2}  a_{j_2}^2 
		\ldots \sum\limits_{j_{k-1}}  a_{j_{k-1}}^{k-1}
		\sum\limits_i  a_{i - j_1 - \ldots - j_{k-1}}^k \cdot \deltas{s_i}.
	\end{aligned}
\end{equation*}	
Thus
\begin{equation} \label{Aeq:C31}
	\begin{aligned}
	& (\widetilde{\omega}_1 \circledast \cdots  \circledast\widetilde{\omega}_k - 
	\omega_1 \circledast \cdots \circledast \omega_k)(s)  \\	
& \quad \quad \quad \quad	=  \sum\limits_{j_1}  a_{j_1}^1 \sum\limits_{j_2}  a_{j_2}^2
		\ldots \sum\limits_{j_{k-1}} a_{j_{k-1}}^{k-1} 
		\sum\limits_i  \widehat{a}_{i - j_1 - \ldots - j_{k-1}}^k \cdot \deltas{s_i},
	\end{aligned}
\end{equation}	
where 
\begin{equation} \label{Aeq:C32}
	\widehat{a}^k_{j} = \widetilde{a}^k_j - a^k_j = \begin{cases}
		 0 , &\text{ if } -L \leq s_j < L, \\
				a^k_{j \, \textrm{mod} \, n}, &\text{ else. }
	\end{cases}
\end{equation}
From \eqref{Aeq:C31} we see that 
\begin{equation} \label{Aeq:C34}
	\begin{aligned}
 & \int\limits_\veps^L\abs{(\omega_1 \circledast \cdots \circledast \omega_k - \widetilde{\omega}_1 \circledast \cdots  \circledast\widetilde{\omega}_k)(s)} \, \dd s  \\
  &  \quad \quad \leq \int\limits_{\mathbb{R}} \abs{(\omega_1 \circledast \cdots \circledast \omega_k - \widetilde{\omega}_1 \circledast \cdots  \circledast\widetilde{\omega}_k)(s)} \, \dd s \\
 & \quad \quad = \int\limits_{\mathbb{R}} \sum\limits_{j_1}  a_{j_1}^1 \sum\limits_{j_2}  a_{j_2}^2
		\ldots \sum\limits_{j_{k-1}} a_{j_{k-1}}^{k-1} 
		\sum\limits_i  \widehat{a}_{i - j_1 - \ldots - j_{k-1}}^k \cdot \deltas{s_i} \, \dd s \\
 & \quad \quad = \sum\limits_{j_1} a_{j_1} \sum\limits_{j_1}  a_{j_1}^1 \sum\limits_{j_2}  a_{j_2}^2
		\ldots \sum\limits_{j_{k-1}} a_{j_{k-1}}^{k-1} 
		\sum\limits_i  \widehat{a}_{i - j_1 - \ldots - j_{k-1}}^k.
	\end{aligned}
\end{equation}
From the periodic form of the coefficients $\widehat{a}^k_{j}$ \eqref{Aeq:C32} we have that
\begin{equation} \label{Aeq:C35}
	\begin{aligned}
& \sum\limits_{j_1}  a_{j_1}^1 \sum\limits_{j_2}  a_{j_2}^2
		\ldots \sum\limits_{j_{k-1}} a_{j_{k-1}}^{k-1} 
		\sum\limits_i  \widehat{a}_{i - j_1 - \ldots - j_{k-1}}^k \\
		& \quad \quad =  \sum\limits_{n \in \mathbb{Z}\setminus \{0\}} \mathbb{P} \big( (2 n - 1) L \leq  \omega_{1} * \cdots * \omega_{k} < (2n+1) L \big) \\
		& \quad \quad =  \sum\limits_{n \in \mathbb{Z}^-} \mathbb{P} \big( (2 n - 1) L \leq  \omega_{1} * \cdots * \omega_{k} < (2n+1) L \big)  \\
		 & \quad \quad \quad \quad + \sum\limits_{n \in \mathbb{Z}^+} \mathbb{P} \big( (2 n - 1) L \leq  \omega_{1} * \cdots * \omega_{k} < (2n+1) L \big) \\
		& \quad \quad \leq  \sum\limits_{n \in \mathbb{Z}^-} \mathbb{P} \big( \omega_{1} * \cdots * \omega_{k} \leq (2n+1) L \big)
		+ \sum\limits_{n \in \mathbb{Z}^+} \mathbb{P} \big(  \omega_{1} * \cdots * \omega_{k}  \geq (2 n - 1) L \big). \\
	\end{aligned}
\end{equation}
We also see that
$$
\sum\limits_{n \in \mathbb{Z}^-} \mathbb{P} \big( \omega_{1} * \cdots * \omega_{k} \leq (2n+1) L \big) 
= \sum\limits_{n \in \mathbb{Z}^+} \mathbb{P} \big( (-\omega_1) * \cdots * (-\omega_k) \geq (2n-1) L \big).
$$
Using the Chernoff bound \eqref{Aeq:chernoff}, we have the tail bound
\begin{equation*}
	\begin{aligned}
\mathbb{P} \big(  \omega_{1} * \cdots * \omega_{k}  \geq (2 n - 1) L \big)  & \leq \ee^{\sum_{i=1}^k \alpha_i(\lambda)}  \ee^{- (2n - 1) L \lambda} \\
& = \ee^{\alpha^+(\lambda)} \ee^{- (2n - 1) L \lambda}
	\end{aligned}
\end{equation*}
and similarly 
$$
\mathbb{P} \big( (-\omega_1) * \cdots * (-\omega_k)  \geq (2 n - 1) L \big) 
 \leq \ee^{\alpha^-(\lambda)} \ee^{- (2n - 1) L \lambda}.
$$
Using the bounds \eqref{Aeq:C34}, \eqref{Aeq:C35} and the Chernoff bound \eqref{Aeq:chernoff}, we find that for all
$\lambda > 0$
\begin{equation*}
	\begin{aligned}
		\int\limits_\veps^L\abs{(\omega_{1} * \cdots * \omega_{k} - \widetilde{\omega}_1 \circledast \cdots  \circledast\widetilde{\omega}_k)(s)} \, \dd s
		 & \leq  \sum\limits_{\ell=1}^\infty \ee^{ \alpha^+(\lambda)} \ee^{- (2 \ell - 1) L \lambda} + \ee^{ \alpha^-(\lambda)} \ee^{- (2 \ell - 1) L \lambda} \\
		 & = \big(  \ee^{ \alpha^+(\lambda)} + \ee^{ \alpha^-(\lambda)}  \big) \frac{ \ee^{- L \lambda} }{ 1 - \ee^{- 2 L \lambda}}.
	\end{aligned}
\end{equation*}
\end{proof}
\end{thm}

\subsection{Thm. 7 of the main text, support of PLD outside of $[-L,L]$} \label{Asubsec:trunc_analysis}

For completeness, we consider also the case, where the PLD distribution is not contained in the $[-L,L]$-interval.
This mean that, in addition to the periodisation error, we give also a bound for the truncation error (which does not vanish in this case)
$$
\int\limits_\veps^L \abs{(\omega_1 * \cdots * \omega_k - \omega_1 \circledast \cdots \circledast \omega_k)(s)} \, \dd s
$$ 
in terms of the moment generating function of $\omega$.



\begin{thm} \label{Athm:trunc}
Let $\omega_i$'s, $\alpha_i^+(\lambda)$'s and $\alpha_i^-(\lambda)$'s be defined as above. For all $\lambda > 0$, we have that
\begin{equation*}
	\begin{aligned}
I_2(L) & = \int\limits_{\veps}^L (\omega_1 * \cdots * \omega_k - \omega_1 \circledast \cdots \circledast \omega_k)(s) \, \dd s \\
& \leq \bigg( \frac{\ee^{k \max_i \alpha_i^+(\lambda)}-\ee^{ \max_i \alpha_i^+(\lambda)}}{\ee^{\max_i \alpha_i^+(\lambda)}-1}
		+ \frac{\ee^{k \max_i \alpha_i^-(\lambda)}-\ee^{ \max_i \alpha_i^-(\lambda)}}{\ee^{\max_i \alpha_i^-(\lambda)}-1} \bigg)  \, \ee^{-L \lambda}.
	\end{aligned}
\end{equation*}
\begin{proof}
By adding and subtracting $(\omega_1 * \cdots * \omega_{k-1}) \circledast \omega_k $ , we may write
\begin{equation} \label{Aeq:C21}
	\begin{aligned}
		& \omega_1 * \cdots * \omega_k - \omega_1 \circledast \cdots \circledast \omega_k \\
		= &	(\omega_1 * \cdots * \omega_{k-1}) * \omega_k - (\omega_1 * \cdots * \omega_{k-1}) \circledast \omega_k \\
		& \quad +  ( \omega_1 * \cdots * \omega_{k-1} - \omega_1 \circledast \cdots \circledast \omega_{k-1} ) \circledast \omega_k.
	\end{aligned}
\end{equation}
Let $\ell \in \mathbb{Z}^+$. Let $\omega_\ell$ be of the form 
$$
\omega_\ell(s) = \sum_j a_j^\ell \cdot \deltas{s_j}
$$
and let the convolution $\omega * \cdots * \omega_{\ell-1}$ be of the form 
$$
(\omega * \cdots * \omega_{\ell-1})(s) = \sum_j c_j \cdot \deltas{s_j}
$$
for some $a^\ell_j,c_j \geq 0$, $s_j = j \Delta x$. 
From the definition of the operators $*$ and $\circledast$ it follows that
\begin{equation*}
	\begin{aligned}
		 & \big( (\omega * \cdots * \omega_{\ell-1}) * \omega_\ell - (\omega * \cdots * \omega_{\ell-1}) \circledast \omega_\ell \big) (s) \\
		& \quad \quad =   \sum\limits_{j_1} \Big( \sum\limits_{j_2}  c_{j_2} a_{j_1-j_2}   \Big) \cdot \deltas{s_{j_1}}
		- \sum\limits_{j_1} \Big( \sum\limits_{-L \leq s_{j_2} < L}  c_{j_2} a_{j_1-j_2}   \Big) \cdot \deltas{s_{j_1}} \\
		& \quad \quad =   \sum\limits_{j_1} \Big( \sum\limits_{ s_{j_2} < -L, \, s_{j_2} \geq L}  c_{j_2} a_{j_1-j_2}   \Big) \cdot \deltas{s_{j_1}}.
	\end{aligned}
\end{equation*}
Therefore,
\begin{equation} \label{Aeq:C22}
	\begin{aligned}
& \int\limits_{\mathbb{R}} \big( (\omega * \cdots * \omega_{\ell-1}) * \omega_\ell - (\omega * \cdots * \omega_{\ell-1}) \circledast \omega_\ell \big) (s) \, \dd s \\
	& \quad \quad =  \int\limits_{\mathbb{R}} \sum\limits_{j_1} \Big( \sum\limits_{ s_{j_2} < -L, \, s_{j_2} \geq L}  c_{j_2} a_{j_1-j_2}^\ell   \Big) \cdot \deltas{s_{j_1}} \, \dd s \\
	& \quad \quad = \sum\limits_{ s_{j_2} < -L, \, s_{j_2} \geq L}  c_{j_2}    \sum\limits_{j_1} \int\limits_{\mathbb{R}}  a_{j_1-j_2}^\ell  \cdot \deltas{s_{j_1}} \, \dd s \\
	& \quad \quad = \sum\limits_{ s_{j_2} < -L, \, s_{j_2} \geq L}  c_{j_2}    \sum\limits_{j_1} a_{j_1-j_2}^\ell  \\
	& \quad \quad =  \sum\limits_{ s_{j_2} < -L, \, s_{j_2} \geq L}  c_{j_2} \\
    & \quad \quad =  \, \mathbb{P} \Big( \omega_1 * \cdots * \omega_{\ell-1} < -L \Big) +  \mathbb{P} \Big( \omega_1 * \cdots * \omega_{\ell-1} \geq L \Big)  \\
	& \quad \quad \leq \, \ee^{ \sum\nolimits_{i=1}^{\ell-1}\alpha_i^+(\lambda)} \ee^{-L \lambda} + \ee^{\sum\nolimits_{i=1}^{\ell-1}\alpha_i^-(\lambda)} \ee^{-L \lambda} \\
	& \quad \quad \leq \, \ee^{ (\ell-1) \max_i \alpha_i^+(\lambda)} \ee^{-L \lambda} + \ee^{(\ell-1) \max_i \alpha_i^-(\lambda)} \ee^{-L \lambda}
	\end{aligned}
\end{equation}
for all $\lambda>0$. The second last inequality follows from the Chernoff bound.

Similarly, suppose $\omega_1 * \cdots * \omega_{\ell-1} - \omega_1 \circledast \cdots \circledast \omega_{\ell-1}$
is of the form 
$$
(\omega_1 * \cdots * \omega_{\ell-1} - \omega_1 \circledast \cdots \circledast \omega_{\ell-1})(s) = \sum_i \widetilde{c}_i \cdot \deltas{s_i}
$$
for some coefficients $\widetilde{c}_i \geq 0$, where $s_i = i \Delta x$. Then,
\begin{equation} \label{Aeq:C23}
	\begin{aligned}
	&  \int\limits_{\mathbb{R}} \big( ( \omega_1 * \cdots * \omega_{\ell-1} - \omega_1 \circledast \cdots \circledast \omega_{\ell-1} ) \circledast \omega_\ell \big) (s) \, \dd s \\
	& \quad \quad =  \int\limits_{\mathbb{R}} \sum\limits_{j_1} \Big( \sum\limits_{-L \leq s_{j_2} < L}  \widetilde{c}_j a_{{j_1}-{j_2}}^\ell   \Big) \cdot \deltas{s_{j_1}} \, \dd s \\
	& \quad \quad =  \sum\limits_{-L \leq s_{j_2} < L}  \widetilde{c}_{j_2} \int\limits_{\mathbb{R}} \sum\limits_i a_{{j_1}-{j_2}}^\ell   \cdot \deltas{s_{j_1}} \, \dd s \\
	& \quad \quad =  \sum\limits_{-L \leq s_{j_2} < L}  \widetilde{c}_{j_2} \\
	& \quad \quad \leq \int\limits_{\mathbb{R}} ( \omega_1 * \cdots * \omega_{\ell-1} - \omega_1 \circledast \cdots \circledast \omega_{\ell-1} )(s) \, \dd s.
	\end{aligned}
\end{equation}
Using \eqref{Aeq:C21}, \eqref{Aeq:C22} and \eqref{Aeq:C23}, we see that for all $\lambda > 0$,
\begin{equation} \label{Aeq:C24}
	\begin{aligned}
		& \int\limits_{\veps}^L (\omega_1 * \cdots * \omega_k - \omega_1 \circledast \cdots \circledast \omega_k)(s) \, \dd s \\
		& \quad \quad \leq   \int\limits_{\mathbb{R}} (\omega_1 * \cdots * \omega_k - \omega_1 \circledast \cdots \circledast \omega_k)(s) \, \dd s  \\
		& \quad \quad \leq   \ee^{ (k-1) \max_i \alpha_i^+(\lambda)} \ee^{-L \lambda} + \ee^{(k-1) \max_i \alpha_i^-(\lambda)} \ee^{-L \lambda} \\
		& \quad \quad \quad \quad   +		\int\limits_{\mathbb{R}} ( \omega_1 * \cdots * \omega_{k-1} - \omega_1 \circledast \cdots \circledast \omega_{k-1} \omega )(s) \, \dd s.
	\end{aligned}
\end{equation}
Using \eqref{Aeq:C24} recursively, we see that for all $\lambda > 0$, we have
\begin{equation*}
	\begin{aligned}
		& \int\limits_{\veps}^L (\omega_1 * \cdots * \omega_k - \omega_1 \circledast \cdots \circledast \omega_k)(s) \, \dd s \\
	 	& \quad \quad \leq  \sum\limits_{\ell=1}^{k-1} \ee^{ \ell \max_i \alpha_i^+(\lambda)} \ee^{-L \lambda}
		+ \sum\limits_{\ell=1}^{k-1} \ee^{ \ell \max_i \alpha_i^-(\lambda)} \ee^{-L \lambda} \\  
		& \quad \quad =  \bigg( \frac{\ee^{k \max_i \alpha_i^+(\lambda)}-\ee^{ \max_i \alpha_i^+(\lambda)}}{\ee^{\max_i \alpha_i^+(\lambda)}-1}
		+ \frac{\ee^{k \max_i \alpha_i^-(\lambda)}-\ee^{ \max_i \alpha_i^-(\lambda)}}{\ee^{\max_i \alpha_i^-(\lambda)}-1} \bigg) \, \ee^{-L \lambda}.
	\end{aligned}
\end{equation*}
\end{proof}
\end{thm}

\subsection{Theorem 8 of the main text}

We next give a proof for Theorem 8 of the main text which gives a bound for the grid approximation error.

\begin{thm}[Thm. 8 of the main text] \label{Alem:discretisation_error}
The discretisation error $\delta^\mathrm{R}(\veps) - \delta(\veps)$ can be bounded as
\begin{equation*} 
	\delta^\mathrm{R}(\veps) - \delta(\veps) \leq   k \Delta x \,  \big( \mathbb{P}( \omega_1 + \cdots + \omega_k \geq \veps ) - \delta(\veps) \big).
\end{equation*}	

\begin{proof}
From the definition of the discrete convolution we see that
\begin{equation*}
	\begin{aligned}
			( \omega_1 * \cdots * \omega_k) (s) = \sum\limits_{i_1,\ldots,i_k} a^1_{i_1} \cdots \, a^k_{i_k} \cdot \delta_{s^1_{i_1} + \ldots + s^k_{i_k}}(s)
	\end{aligned}
\end{equation*}
and that
\begin{equation*}
	\begin{aligned}
	\delta(\veps)	&	= \int\nolimits_\veps^\infty (1 - \ee^{\veps - s} ) \, (\omega_1 * \cdots * \omega_k)(s) \, \dd s \\
	&   =  	\sum\limits_{  \{ i_1,\ldots,i_k \, : \,  s^1_{i_1} + \ldots + s^k_{i_k} \geq \veps \}} 
	a^1_{i_1} \cdots \, a^k_{i_k}  \cdot (1 - \ee^{\veps - s^1_{i_1} - \hdots - s^k_{i_k}} )
	\end{aligned}
\end{equation*}
Then, using the inequality 
$$ 
a\leq b \Rightarrow \exp(b) - \exp(a) \leq \exp(b)(b-a),
$$ 
we have that
\begin{equation*}
	\begin{aligned}
		 \delta^\mathrm{R}(\veps) - \delta(\veps)  
		 = & \sum\limits_{  \{ i_1,\ldots,i_k \, : \,  s^1_{i_1} + \ldots + s^k_{i_k} \geq \veps \}} a^1_{i_1} \cdots \, a^k_{i_k}  \cdot 
		   (\ee^{\veps - s^1_{i_1} - \hdots - s^k_{i_k}} - \ee^{\veps - s_{i_1}^{\mathrm{R},1} - \hdots - s_{i_k}^{\mathrm{R},k}} ) \\
		\leq  &\sum\limits_{  \{ i_1,\ldots,i_k \, : \,  s^1_{i_1} + \ldots + s^k_{i_k} \geq \veps \}} a^1_{i_1} \cdots \, a^k_{i_k} \cdot
		   \big( (s_{i_1}^{\mathrm{R},1} - s^1_{i_1}) + \hdots + (s_{i_k}^{\mathrm{R},k} - s^k_{i_k})\big) \cdot \ee^{\veps - s^1_{i_1} - \hdots - s^k_{i_k}}.
	\end{aligned}
\end{equation*}
Since $s_i^{\mathrm{R},\ell} - s^\ell_i \leq \Delta x $ for all $i$ and $\ell$, we have that
\begin{equation*} 
	\begin{aligned}
		 \delta^\mathrm{R}(\veps) - \delta(\veps)  
		&\leq   \, k \Delta x \sum\limits_{  \{ i_1,\ldots,i_k \, : \,  s^1_{i_1} + \ldots + s^k_{i_k} \geq \veps \}}  a^1_{i_1} \cdots \, a^k_{i_k} \cdot  \ee^{\veps - s^1_{i_1} - \hdots - s^k_{i_k}} \\
		& =   \, k \Delta x \big(\sum\limits_{  \{ i_1,\ldots,i_k \, : \,  s^1_{i_1} + \ldots + s^k_{i_k} \geq \veps \}}  a^1_{i_1} \cdots \, a^k_{i_k} \\
		&  -  \sum\limits_{  \{ i_1,\ldots,i_k \, : \,  s^1_{i_1} + \ldots + s^k_{i_k} \geq \veps \}}  a^1_{i_1} \cdots \, a^k_{i_k} \cdot (1- \ee^{\veps - s^1_{i_1} - \hdots - s^k_{i_k}} )  \big) \\
		& = k \Delta x \,  \big( \mathbb{P}( \omega_1 * \cdots * \omega_k \geq \veps ) - \delta(\veps) \big).
	\end{aligned}
\end{equation*}
\end{proof}
\end{thm}

\subsection{Lemma 10 the Main Text: Upper Bound for the Computational Complexity}

We next prove the computational complexity result of Lemma 10 of the main text.
For simplicity, we assume in the following that the compositions consist of $(\veps,0)$-DP mechanisms and that the parameter 
$L$ is chose sufficiently large so that for all $i$: $\abs{s_i} \leq L$, where $s_i = \log \tfrac{a_{X,i}}{a_{Y,i}}$.
Then, we can bound the periodisation error using Theorem 7 of the main text. 

Recall: we denote the logarithms of the moment generating functions of the PLDs as
\begin{equation} \label{eq:alphas_def}
	\alpha_i^+(\lambda) = \log \Big(	\mathbb{E} [\ee^{\lambda \omega_i }] \Big),
	\quad  \alpha_i^-(\lambda) = \log \Big(	\mathbb{E} [\ee^{ - \lambda \omega_i }] \Big),
\end{equation}
where $1 \leq i \leq k$. Futhermore, denote 
\begin{equation*}
	\alpha^+(\lambda) = \sum\nolimits_i \, \alpha^+_i(\lambda),
	\quad  \alpha^-(\lambda) = \sum\nolimits_i \, \alpha^-_i(\lambda).
\end{equation*}

\begin{lem}
Consider a non-adaptive composition of the mechanisms $\mathcal{M}_1, \ldots, \mathcal{M}_k$ with corresponding worst-case pairs of distributions
$f_{X,i}$ and $f_{Y,i}$, $1 \leq i \leq k$. Suppose the sequence $\mathcal{M}_1, \ldots, \mathcal{M}_k$ consists of $m$ distinct mechanisms.
Then, it is possible to have an approximation of $\delta(\veps)$ with error less than $\eta$ with number of operations
$$
\mathcal{O}\left(\frac{m k^2 C_k}{\eta} \log  \frac{k^2 C_k}{\eta} \right),
$$
where 
$$
C_k = \max \{ \tfrac{1}{k} \sum_i D_{\infty}(f_{X,i} || f_{Y,i}),  \tfrac{1}{k} \sum_i D_{\infty}(f_{Y,i} || f_{X,i}) \}
$$
and
$$
D_{\infty}(f_X || f_Y) = \sup_{a_{Y,i} \neq 0} \log \frac{a_{X,i}}{a_{Y,i}}.
$$
\begin{proof}

We first determine a lower bound for the truncation parameter $L$ in terms of $k$.
Consider the right-hand-side of the error bound of Theorem 7 of the main text. Suppose $L \geq 1$ and $\lambda \geq 1$. 
Then, we have that
\begin{equation} \label{Aeq:ub1}
	\begin{aligned}
 \big(  \ee^{\alpha^+(\lambda)} + \ee^{ \alpha^-(\lambda)} \big) \, \frac{\ee^{-L \lambda}}{1- \ee^{-2 L \lambda}} 
 & \leq   \big(  \ee^{\alpha^+(\lambda)} + \ee^{ \alpha^-(\lambda)} \big) \cdot \frac{\ee}{2} \cdot \ee^{-L \lambda}, \\
 & \leq   \ee^{\max \{  {\alpha^-(\lambda), \alpha^+(\lambda)} \} + 1} \ee^{-L \lambda},
	\end{aligned}
\end{equation}
where $\alpha^-(\lambda)$ and $\alpha^+(\lambda)$ are defined as above in~\eqref{eq:alphas_def}.

For each $i$, the logarithm of the moment-generating function of the PLD
can be expressed in terms of the R\'enyi divergence~\cite{mironov2017}:
\begin{equation*}
	\begin{aligned}
		\log \Big(	\mathbb{E} [\ee^{\lambda \omega_{X/Y,i} }] \Big)
		& = \lambda \cdot \frac{1}{\lambda} \sum_i \left( \frac{a_{X,i}}{a_{Y,i}}  \right)^\lambda a_{X,i} \\
		& = \lambda \cdot \frac{1}{\lambda} \sum_i \left( \frac{a_{X,i}}{a_{Y,i}}  \right)^{\lambda+1} a_{Y,i} \\
		& = \lambda \cdot D_{\lambda+1}(f_{X} || f_{Y}),		
	\end{aligned}
\end{equation*}
where $D_{\lambda}$ denotes the R\'enyi divergence of order $\lambda$.
From the monotonicity of R\'enyi divergence (see Proposition 9, \cite{mironov2017}) it follows that 
\begin{equation*}
	\begin{aligned}
\alpha^+(\lambda) & = \lambda \cdot \sum\nolimits_i D_{\lambda+1}(f_{X,i} || f_{Y,i}) \\
 & \leq \lambda \cdot \sum\nolimits_i D_{\infty}(f_{X,i} || f_{Y,i}),
	\end{aligned}
\end{equation*}
where 
$$
D_{\infty}(f_X || f_Y) = \sup_{a_{Y,i} \neq 0} \log \frac{a_{X,i}}{a_{Y,i}}.
$$
With a similar calculation, we find that 
\begin{equation*}
	\begin{aligned}
 \alpha^-(\lambda) \leq (\lambda-1) \cdot \sum\nolimits_i D_{\infty}(f_{Y,i} || f_{X,i}).
	\end{aligned}
\end{equation*}
Thus,
\begin{equation*}
	\begin{aligned}
 \max \{  {\alpha^-(\lambda), \alpha^+(\lambda)} \} \leq 
 k \lambda \cdot \max \{ \tfrac{1}{k} \sum_i D_{\infty}(f_{X,i} || f_{Y,i}),  \tfrac{1}{k} \sum_i D_{\infty}(f_{Y,i} || f_{X,i}) \}. 
	\end{aligned}
\end{equation*}
Now we can further bound \eqref{Aeq:ub1} from above as 
$$
 \ee^{\max \{  {\alpha^-(\lambda), \alpha^+(\lambda)} \}+ 1} \ee^{-L \lambda} \leq  \ee^{ k \lambda \cdot C_k + 1 } \ee^{-L \lambda},
$$
where 
$$
C_k = \max \{ \tfrac{1}{k} \sum_i D_{\infty}(f_{X,i} || f_{Y,i}),  \tfrac{1}{k} \sum_i D_{\infty}(f_{Y,i} || f_{X,i}) \}.
$$
Requiring this upper bound to be smaller than a prescribed $\eta > 0$, and setting $\lambda = 1$, we arrive at the condition
\begin{equation} \label{Aeq:cond_lower_L}
	L \geq k \cdot C_k + 1 + \log \frac{1}{\eta}.
\end{equation}

Next, we bound the computational complexity using a bound for the discretisation error.
From Thm. 8 of the main text it follows that the discretisation error is bounded as
$$
\delta^\mathrm{R}(\veps) - \delta(\veps) \leq k \Delta x = \frac{2 L k}{n}.
$$
Requiring this discretisation error to be less than $\eta$,
choosing $L$ according to \eqref{Aeq:cond_lower_L} and assuming $k \geq \log \frac{1}{\eta}$, we see that choosing
$$
n = \mathcal{O}\left(\frac{k^2 C_k }{\eta} \right)
$$
is sufficient for the sum of the error sources to be less than $2 \eta$. As we need to compute FFT for $m$ different PLDs, and since FFT has complexity $n \log n$, we see that 
with 
$$
\mathcal{O}\left(\frac{2 m k^2 C_k }{\eta} \log  \frac{ k^2 C_k }{\eta} \right)
$$
operations it is possible to have an approximation of $\delta(\veps)$ with error less than $\eta$,
and that additional factor in the leading constant is given by the leading constant in the complexity of FFT.

\end{proof}
\end{lem}

\begin{remark}
We see from the proof, that for the condition that the periodisation error is less than $\eta$, it is sufficient to choose
$$
L \geq \frac{\log \eta^{-1}  + \max \{  {\alpha^-(\lambda), \alpha^+(\lambda)} \} + 1 }{\lambda}.
$$
As this is true for all $\lambda \geq 1$ and since $\alpha^-(\lambda)$ and $\alpha^+(\lambda)$ can be evaluated numerically, a sufficient value of $L$
can be found via an optimisation problem. Notice also that since $\alpha^-(\lambda)$ and $\alpha^+(\lambda)$ correspond to cumulant generating functions (CGFs)~\citep{Abadi2016},
and since the minimisation problem 
$$
\min_\lambda \frac{\log \delta^{-1}  +  \alpha(\lambda) }{\lambda}
$$
is exactly the conversion formula for turning CGF-values to $(\veps(\delta),\delta)$-DP values
, we see that approximately (assuming $\lambda^{-1}$ is small)
$L$ has to be chosen as
$$
L \geq \veps(\eta),
$$
where $\veps(\eta)$ gives $(\veps,\delta)$-DP of the composition $(\mathcal{M}_1, \ldots, \mathcal{M}_k)$ at $\delta=\eta$. 

\end{remark}

\subsection{Fast Evaluation Using the Plancherel Theorem}

We next prove Lemma 11 of the main text. Recall the Fourier accountant algorithm of the main text.
When using this algorithm to approximate $\delta(\veps)$, we need to evaluate the expression
\begin{equation} \label{Aeq:b_k_p}
\boldsymbol{b}^k = D \, \mathcal{F}^{-1} \big(\mathcal{F}( D \boldsymbol{a} )^{\odot k}   \big)
\end{equation}
and the sum
\begin{equation} \label{Aeq:summ}
\widetilde{\delta}(\veps) = \sum\nolimits_{ - L + \ell \Delta x > \veps}  \big(1 - \ee^{\veps - ( - L + \ell \Delta x)} \big) \, b^k_\ell.
\end{equation}
The following lemma shows that when evaluating $\widetilde{\delta}(\veps)$ for different numbers of compositions $k$,
the updates of $\widetilde{\delta}(\veps)$ can be performed in $\mathcal{O}(n)$ time.
\begin{lem} \label{Alem:plancherel}
Denote $\boldsymbol{w}_\veps \in \mathbb{R}^n$ such that
\begin{equation*}
	(\boldsymbol{w}_\veps)_\ell = \max\{  1 - \ee^{\veps - ( - L + \ell \Delta x)} , 0\}
\end{equation*}
and let $\boldsymbol{b}^k$ be of the form \eqref{Aeq:b_k_p}.
Then, we have that
\begin{equation} \label{AAeq:planch0}
\widetilde{\delta}(\veps) = \frac{1}{n} \langle \mathcal{F}( D \boldsymbol{w}_\veps ) , \mathcal{F}( D \boldsymbol{a} )^{\odot k}   \rangle.
\end{equation}
\begin{proof}
We see that the sum \eqref{Aeq:summ} is given by the following inner product:
$$
\widetilde{\delta}(\veps) = \langle \boldsymbol{w}_\veps, \boldsymbol{b}^k\rangle.
$$
The Plancherel Theorem states that the discrete Fourier transform preserves inner products: for $x,y \in \mathbb{R}^n$,
\begin{equation} \label{Aeq:planch}
	\langle x, y \rangle = \frac{1}{n} \langle \mathcal{F} x, \mathcal{F} y \rangle.
\end{equation}
Using \eqref{Aeq:planch}, we see that
\begin{equation*}
	\begin{aligned}
	\widetilde{\delta}(\veps) &= \langle \boldsymbol{w}_\veps, \boldsymbol{b}^k\rangle \\
	  &= \langle \boldsymbol{w}_\veps, D \mathcal{F}^{-1}\big(  \mathcal{F}( D \boldsymbol{a} )^{\odot k} \big) \rangle \\
	  &= \langle D \boldsymbol{w}_\veps, \mathcal{F}^{-1}\big(  \mathcal{F}( D \boldsymbol{a} )^{\odot k} \big) \rangle \\
	  &= \frac{1}{n} \langle \mathcal{F}( D \boldsymbol{w} ) , \mathcal{F}( D \boldsymbol{a} )^{\odot k}   \rangle.
	\end{aligned}
\end{equation*}
\end{proof}
\end{lem}

\section{Details for the Experiments of Section 6} 

In order to use the Fourier accountant for compositions including continuous mechanisms, we first need to discretise
the PLDs of the continuous mechanism appropriately. This means that we replace each continuous PLD $\omega$
by a certain discrete-valued distribution $\omega_{\mathrm{max}}$ that leads to an overall $\delta(\veps)$-upper bound.
This procedure is analogous to what is considered in the experiments of~\cite{koskela2021tight} for the homogeneous composition
of subsampled Gaussian mechanisms~\cite[see the supplementary material of][]{koskela2021tight}.  
Those results can be used to derive the discrete PLDs for the experiments of Sec. 6.2.,
i.e., for the heterogeneous composition of subsampled Gaussian mechanisms.

\subsection{Experiments of Section 6.1 }

For the PLD $\omega_\mathrm{G}$ of the Gaussian mechanism we know that~\citep{sommer2019privacy}
$$
\omega_\mathrm{G} ~ \sim \mathcal{N}\left( \frac{1}{2 \sigma^2}, \frac{1}{\sigma^2} \right).
$$
Let $L>0$, $n \in \mathbb{Z}^+$, $\Delta x = 2L/n$ and let the grid $X_n$ be defined as in Section 4.2 of the main text.
Define
\begin{equation*} 
	\begin{aligned}
				\omega_{\mathrm{max}}(s) = \sum\limits_{i=0}^{n-1} c^+_i \cdot \deltas{s_i},
	\end{aligned}
\end{equation*}
where $s_i = i \Delta x$ and 
\begin{equation} \label{eq:c_plusminus}
	\begin{aligned}
		c^+_i = \Delta x \cdot \max\limits_{s \in [s_{i-1}, s_i]} \omega_\mathrm{G}(s),
	\end{aligned}
\end{equation}
and define
\begin{equation} \label{eq:c_plusminus_inf}
	\begin{aligned}
			\omega^\infty_{\mathrm{max}}(s) = \sum\limits_{i \in \mathbb{Z}} c^+_i \cdot \deltas{s_i}.
	\end{aligned}
\end{equation}

To obtain rigorous $\delta(\veps)$-bounds for the compositions, we carry out the error analysis for the distribution
$\omega^\infty_{\mathrm{max}}$ and use Theorem~\ref{Athm:trunc} above.
To this end, we need bounds for
the moment generating functions of $-\omega^\infty_{\mathrm{max}}$ and $\omega^\infty_{\mathrm{max}}$. 

To show that $\omega^\infty_{\mathrm{max}}$ indeed leads to an upper bound for $\delta(\veps)$,
we refer to~\citep[supplementary material of][]{koskela2021tight}, where this is shown for the compositions of the subsampled Gaussian mechanism. The proof here goes analogously,
and we have that for all $\veps>0$,
$$
\delta(\veps) \leq \delta_{\mathrm{max}}^\infty(\veps),
$$
where $\delta_{\mathrm{max}}^\infty(\veps)$ is the tight bound for the composition involving $\omega^\infty_{\mathrm{max}}$.

To evaluate $\alpha^+(\lambda)$ and $\alpha^-(\lambda)$ for the upper bound of Theorem~\ref{Athm:trunc}, 
we need the moment generating functions of 
$-\omega^\infty_{\mathrm{max}}$ and $\omega^\infty_{\mathrm{max}}$. We have the following bound for $\omega^\infty_{\mathrm{max}}$.
 We note that $\mathbb{E} [\ee^{\lambda \omega_{\mathrm{max}} }]$ can be evaluated numerically.

\begin{lem} \label{lem:mgfs}
Let $0 < \lambda \leq L$ and assume $\sigma \geq 1$ and $\Delta x \leq c \cdot L$, $0<c<1$.
The moment generating function of 
$\omega^\infty_{\mathrm{max}}$ can be bounded as
$$
\mathbb{E} [\ee^{\lambda \omega^\infty_{\mathrm{max}} }] \leq 
\mathbb{E} [\ee^{\lambda \omega_{\mathrm{max}} }] + \mathrm{err}(\lambda,L,\sigma),
$$
where
\begin{equation} \label{eq:err_term}
	\mathrm{err}(\lambda,L,\sigma) = \exp\bigg(\frac{3\lambda}{2\sigma^2}\bigg) \left(  \int_{-\infty}^{-L} \widetilde{\omega}(s) \, \dd s + \int_{L- \Delta x}^\infty\widetilde{\omega}(s) \, \dd s     \right),
	\quad \widetilde{\omega} ~ \sim \mathcal{N}\left( \frac{1+2 \lambda}{2 \sigma^2}, \frac{1}{\sigma^2} \right).
\end{equation}

\begin{proof}

The moment generating function of $\omega^\infty_{\mathrm{max}}$ is given by
\begin{equation} \label{eq:pld_lmf}
	\begin{aligned}
		\mathbb{E} [\ee^{\lambda \omega^\infty_{\mathrm{max}} }] &= \int_{-L}^L \ee^{\lambda s} \omega^\infty_{\mathrm{max}}(s) \, \dd s 
		+ \int_{-\infty}^{-L} \ee^{\lambda s} \omega^\infty_{\mathrm{max}}(s) \, \dd s + \int_{L}^\infty \ee^{\lambda s} \omega^\infty_{\mathrm{max}}(s) \, \dd s  \\
		&\leq \mathbb{E} [\ee^{\lambda \omega_{\mathrm{max}} }]
		+ \int_{-\infty}^{-L} \ee^{\lambda s} \omega_\mathrm{G}(s) \, \dd s + \int_{L- \Delta x}^\infty \ee^{\lambda s} \omega_\mathrm{G}(s) \, \dd s  \\
	\end{aligned}
\end{equation}
We arrive at the claim by observing that for $\omega_\mathrm{G} ~ \sim \mathcal{N}\left( \frac{1}{2 \sigma^2}, \frac{1}{\sigma^2} \right)$,
$$
\int_{-\infty}^{-L} \ee^{\lambda s} \omega_\mathrm{G}(s) \, \dd s = \exp\bigg(\frac{3\lambda}{2\sigma^2}\bigg) \int_{-\infty}^{-L} \widetilde{\omega}(s) \, \dd s,
$$
where $\widetilde{\omega} ~ \sim \mathcal{N}\left( \frac{1+2 \lambda}{2 \sigma^2}, \frac{1}{\sigma^2} \right)$ and similarly for the second term in~\eqref{eq:err_term}.
\end{proof}
\end{lem}
Using a reasoning analogous to the proof of Lemma~\ref{lem:mgfs}, we get the following. 
We note that $\mathbb{E} [\ee^{ - \lambda \omega_{\mathrm{min}} }]$ can be evaluated numerically.
\begin{cor}
The moment generating function of $-\omega^\infty_{\mathrm{max}}$ can be bounded as 
$$
\mathbb{E} [\ee^{ - \lambda \omega^\infty_{\mathrm{max}} }] \leq 
\mathbb{E} [\ee^{ - \lambda \omega_{\mathrm{max}} }] + \mathrm{err}(\lambda,L,\sigma),
$$
where
$\mathrm{err}(\lambda,L,\sigma)$ is defined as in \eqref{eq:err_term}.
\end{cor}

\begin{remark}
In the experiments, the error term $\mathrm{err}(\lambda,L,\sigma)$ was found to be negligible.
\end{remark}

\newpage

\section{Tight $(\veps,\delta)$-Bound for an Adaptive Composition of Multivariate Subsampled Gaussian Mechanisms Using One-Dimensional Distributions}

We next give a rigorous proof for the fact that the multivariate subsampled Gaussian mechanism with Poisson subsampling
can be analysed by one-dimensional Gaussian mixtures.
The proof is motivated by an analogous result~\cite[Thm.\;4]{mironov2019} which is for RDP and we partly use
the notation used in the proof of that result. For simplicity, we
focus here on the case, where the underlying function is of the summative form
\begin{equation} \label{eq:f_sum}
	F(X,\theta) = \sum_{x \in X} g(x,\theta),
\end{equation}

where $g$ is Lipschitz-continuous w.r.t. $\theta$ and has a 2-norm bounded by a constant $c>0$.

We show the equivalence part by part so that first in Section~\ref{sec:poisson} we show the equivalence for a single iteration of the mechanism and then
in Section~\ref{sec:adaptive} we show the equivalence rigorously for the multivariate Gaussian mechanism. Then, by combining these 
results arrive at the conclusion.

\subsection{Motivational Example: DP-SGD}

The motivational example to consider functions of the form \eqref{eq:f_sum} is DP-SGD, where the terms $g(\theta,x)$ are the sample-wise clipped gradients. 
When applying DP-SGD to, for example, neural networks, the gradients can be assumed to be Lipschitz continuous
in bounded sets (a condition sufficient for our analysis). As the following result verifies, then are also the clipped gradients
Lipschitz continuous.

\begin{lem}
Suppose the function $h(\theta)$, $h \, : \, S \rightarrow \mathbb{R}^d$ is $L$-Lipschitz continuous in $S \subset \mathbb{R}^d$. Then, also the function 
$$
g(\theta) = \mathrm{clip}_c\big(h\big)(\theta) =  \frac{h(\theta)}{\max\{1,\frac{\norm{h(\theta)}_2}{c} \}}
$$ 
is $L$-Lipschitz in $S$.
\begin{proof}
Let $x,y \in S$.
If $\norm{h(x)}_2 \leq c$ or $\norm{h(y)}_2 \leq c$, or if $\langle h(x),h(y) \rangle < 0$, the inequality
$$
\norm{g(x) - g(y)}_2 \leq \norm{h(x) - h(y)}_2.
$$
follows by simple geometry.
Assume $\norm{h(x)}_2 > c$, $\norm{h(y)}_2 > c$ and	$\langle h(x),h(y) \rangle \geq 0$. Then,
\begin{equation*}
	\begin{aligned}
			\norm{g(x) - g(y)}_2^2 &= \norm{g(x)}_2^2 + \norm{g(y)}_2^2 - 2 \langle g(x),g(y) \rangle \\
			& = 2 c^2 \cdot \big(  1 -  \langle \tfrac{h(x)}{\norm{h(x)}_2},\tfrac{h(y)}{\norm{h(y)}_2} \rangle  \big) \\
			& \leq 2 \norm{h(x)}_2 \norm{h(y)}_2 \cdot \big(  1 -  \langle \tfrac{h(x)}{\norm{h(x)}_2},\tfrac{h(y)}{\norm{h(y)}_2} \rangle  \big) \\
			& =  2 \norm{h(x)}_2 \norm{h(y)}_2 - 2 \langle h(x),h(y) \rangle \\
			& \leq \norm{h(x)}_2^2 + \norm{h(y)}_2^2 - 2 \langle h(x),h(y) \rangle \\
			&= \norm{h(x) - h(y)}_2^2.
	\end{aligned}
\end{equation*}
\end{proof}
\end{lem}

\subsection{The Subsampled Gaussian Mechanism with Poisson Subsampling} \label{sec:poisson}

We first show the analogy for a single iteration of the mechanism,
in a case where the underlying function $g$ is differentiable w.r.t. $\theta$ and has a norm exactly 1 for all data samples $x$ and for all $\theta$.
We will use the following notation repeatedly in the proof.

\begin{defn} \label{def:indistinguishability}
	Let $\varepsilon > 0$ and $\delta \in [0,1]$.
	Let $P$ and $Q$ be two random variables taking values in a measurable space $\mathcal{R}$.
	We say that $P$ and $Q$ are
	$(\varepsilon,\delta)$-indistinguishable, denoted
	$P \simeq_{(\veps,\delta)} Q$,
	if for every measurable set $E \subset \mathcal{R}$ we have
	$$
		\mathrm{Pr}( P \in E ) \leq \ee^\varepsilon \mathrm{Pr} (Q \in E ) + \delta \quad
		\textrm{and} \quad
		\mathrm{Pr}( Q \in E ) \leq \ee^\varepsilon \mathrm{Pr} (P \in E ) + \delta.
	$$
\end{defn}

\begin{thm} \label{thm:subsampled_xyz}
If
$$
	\mathcal{N}( 0,\sigma^2)  \simeq_{(\veps,\delta)} (1-q) \cdot \mathcal{N} \big(	 0,\sigma^2 \big) + q \cdot \mathcal{N}\big( 1 ,\sigma^2  \big),
$$
where $(1-q) \cdot \mathcal{N} (0,\sigma^2 ) + q \cdot \mathcal{N}( 1 ,\sigma^2 )$ denotes a mixture of $\mathcal{N} (0,\sigma^2 )$
and $\mathcal{N} (1,\sigma^2 )$,
then also the multivariate Poisson subsampled Gaussian mechanism with subsampling ratio $q$ and variance $\sigma^2$ and $L_2$-sensitivity $1$ is $(\varepsilon,\delta)$-DP. \\

\begin{proof}

Similarly to the proof of~\cite[Thm.\;4]{mironov2019},
let $T$ denote a set-valued random variable defined by taking a random subset of $X \in \mathcal{X}^N$, 
where each element of $X$ is independently placed in $T$ with probability $q$. 
For simplicity, suppose that $f(T)$ is of the summative form
$$
f(T) = \sum_{x \in T} g(x), 
$$
where $\norm{g(x)}_2=1$ for all $x \in X$.
Conditioned on $T$, the mechanism $\mathcal{M}(X)$ samples from a Gaussian with mean $f(T)$.
Then, $\mathcal{M}(X)$ can be represented as
a mixture
$$
\mathcal{M}(X) = \sum_T p_T \cdot \mathcal{N}( f(T),\sigma^2 I_d),
$$
where the sum denotes mixing of the distributions with the weights $p_T$.

Let $X' \in \mathcal{X}^N$ be a neighbouring dataset such that $X' = X \cup \{ x' \}$. Then, we have
$$ 
\mathcal{M}(X') = \sum_T p_T \cdot \bigg( (1-q) \cdot \mathcal{N} \big(	 f(T),\sigma^2 I_d \big) + q \cdot \mathcal{N}\big( f(T)+g(x'),\sigma^2 I_d \big) \bigg).
$$
\textbf{(Rotation)} Consider an orthogonal matrix $U \in \mathbb{R}^{d \times d}$ 
of the form 
$$
U = \begin{bmatrix}  g(x') & \widetilde{U} \end{bmatrix}, 
$$
where $\widetilde{U} \in \mathbb{R}^{d \times (d-1)}$.
Again, $\widetilde{U}$ can be taken as any $d \times (d-1)$ matrix such that the columns of $U$
give an orthonormal basis of $\mathbb{R}^d$.
Then, in particular, we have that
$$
U^\mathrm{T} \Delta = e_1,
$$
where $e_1 = \begin{bmatrix} 1 & 0 & \ldots & 0 \end{bmatrix}^T$.
We see that the fact that $\mathcal{M}(X) \simeq_{(\veps,\delta)} \mathcal{M}(X')$ is equivalent to the fact that 
$U^T \mathcal{M}(X) \simeq_{(\veps,\delta)} U^T \mathcal{M}(X')$. Clearly,
$$
U^T \mathcal{M}(X) 
\sim \sum_T p_T \cdot  \mathcal{N}( U^T f(T),\sigma^2 I_d)
$$
since $U$ is orthogonal. Similarly
$$
U^T \mathcal{M}(X') \sim \sum_T p_T \cdot \bigg( (1-q) \cdot \mathcal{N} \big(	 U^T f(T),\sigma^2 I_d \big) + q \cdot \mathcal{N}\big( U^T (f(T)+g(x')),\sigma^2 I_d \big) \bigg).
$$
\textbf{(Translation)} Clearly for each subset $T$ of $X$, 
\begin{equation} \label{eq:cond21}
	\mathcal{N}( U^T f(T),\sigma^2 I_d) \simeq_{(\veps,\delta)} (1-q) \cdot 
	\mathcal{N} \big(	 U^T f(T),\sigma^2 I_d \big) + q \cdot \mathcal{N}\big( U^T (f(T)+g(x')),\sigma^2 I_d \big)
\end{equation} 
if and only if 
\begin{equation} \label{eq:cond22}
	\mathcal{N}( 0,\sigma^2 I_d) \simeq_{(\veps,\delta)} (1-q) \cdot \mathcal{N} \big(	 0,\sigma^2 I_d \big) + q \cdot \mathcal{N}\big( U^T g(x'),\sigma^2 I_d \big).
\end{equation}
Since $U^T g(x') = e_1$, and the coordinate-wise noises in the mechanisms of \eqref{eq:cond22} are independent,
\eqref{eq:cond22} holds if and only if 
\begin{equation} \label{eq:cond23}
	\mathcal{N}( 0,\sigma^2) \simeq_{(\veps,\delta)} (1-q) \cdot \mathcal{N} \big(	 0,\sigma^2 \big) + q \cdot \mathcal{N}\big( 1 ,\sigma^2  \big).
\end{equation}
Thus,  \eqref{eq:cond21} holds if and only if  \eqref{eq:cond23} holds.
	
\medskip	
	
Now suppose 
$$
\mathcal{N}( 0,\sigma^2) \simeq_{(\veps,\delta)} (1-q) \cdot \mathcal{N} \big(	 0,\sigma^2 \big) + q \cdot \mathcal{N}\big( 1 ,\sigma^2  \big).
$$
Let $S \subset \mathbb{R}^d$. Using the reasoning above, we have
\begin{equation*}
	\begin{aligned}
		\mathbb{P}( U^T \mathcal{M}(X') \subset S ) &= \mathbb{P}( U^T \mathcal{M}(X') \subset S )     \\
		& = \mathbb{P} \bigg( \sum_T p_T \cdot \big( (1-q) \cdot \mathcal{N} \big(	 U^T f(T),\sigma^2 I_d \big) + q \cdot \mathcal{N}\big( U^T (f(T)+g(x')),\sigma^2 I_d \big) \big) \subset S \bigg)     \\
		& = \sum_T p_T \cdot \mathbb{P}\bigg( \big( (1-q) \cdot \mathcal{N} \big(	 U^T f(T),\sigma^2 I_d \big) + q \cdot \mathcal{N}\big( U^T (f(T)+g(x')),\sigma^2 I_d \big) \big) \subset S \bigg) \\
		& \leq \sum_T p_T  \bigg( \ee^\varepsilon  \mathbb{P}(   \mathcal{N}( U^T f(T),\sigma^2 I_d) \subset S ) + \delta \bigg) \\
		& =  \ee^\varepsilon \sum_T p_T \cdot \mathbb{P} \big(   \mathcal{N}( U^T f(T),\sigma^2 I_d) \subset S \big) + \sum_T p_T \cdot  \delta \\
		& =  \ee^\varepsilon  \mathbb{P} \big( \sum_T p_T \cdot  \mathcal{N}( U^T f(T),\sigma^2 I_d) \subset S \big) +  \delta \\
		& =  \ee^\varepsilon  \mathbb{P} \big( U^T \mathcal{M}(X) \subset S \big) +  \delta. \\
	\end{aligned}
\end{equation*}
Similarly, we see that 
\begin{equation*}
	\begin{aligned}
		\ee^\varepsilon  \mathbb{P} \big( U^T \mathcal{M}(X) \subset S \big) \leq 
		\mathbb{P}( U^T \mathcal{M}(X') \subset S ) +  \delta.
	\end{aligned}
\end{equation*}
Since the fact that $\mathcal{M}(X) \simeq_{(\veps,\delta)} \mathcal{M}(X')$is equivalent to the fact that 
$U^T \mathcal{M}(X) \simeq_{(\veps,\delta)} U^T \mathcal{M}(X')$, we see that 
$\mathcal{M}(X) \simeq_{(\veps,\delta)} \mathcal{M}(X')$.

\end{proof}
\end{thm}

\subsection{Adaptive Compositions}  \label{sec:adaptive}
The PLD approach is directly applicable to non-adaptive compositions of the form 
\begin{equation} \label{eq:comp_non_adaptive}
	\mathcal{M}(X) = \big(\mathcal{M}_1(X), \ldots, \mathcal{M}_k(X) \big).
\end{equation}
The adaptive compositions we consider are of the form
$$
\mathcal{M}(X,\theta) = \big(\mathcal{M}_1(X,\theta), \mathcal{M}_2\big(X,\mathcal{M}_1(X,\theta) \big),\ldots, \mathcal{M}_k\big(X,\mathcal{M}_{k-1}(X,\ldots)\big).
$$
We want to bound the tight $(\veps,\delta)$-values of the adaptive composition with the a non-adaptive composition of the form~\eqref{eq:comp_non_adaptive}.
We first recall an integral representation for the privacy loss random variable that will be of use.


\subsubsection{Representations for Tight DP-Guarantees}


When analysing general DP mechanisms, we use the following definition for the privacy loss random variable.
We write $f_X(t)$, $t \in \mathbb{R}^{k \cdot t}$, for the density function 
of $\mathcal{M}(X,\theta)$ and $f_{X'}(t)$ for the density function of $\mathcal{M}(X',\theta)$.
\begin{defn} \label{defn:pld}
The privacy loss random variable is a measure $\omega \, : \, \mathbb{R} \cup \{ \infty \} \rightarrow [0,1]$, such that
for $S \subset \mathbb{R} \cup \{ \infty \}$,
$$
\omega(S) = \int\limits_{\{ t \in \mathbb{R}^d \, : \, \mathcal{L}_{X/{X'}}(t) \in S   \}} f_X(t) \, \dd t,
$$
where $\mathcal{L}_{X/{X'}}(t) = \log \frac{f_X(t)}{f_{X'}(t)}$ denotes \emph{the privacy loss function}.

\end{defn}

We first recall the following result for mechanisms in $\mathbb{R}^d$:

\begin{lem} \label{lem:tight_d}
$\mathcal{M}(X) \simeq_{(\veps,\delta)} \mathcal{M}({X'})$ (tightly) with
\begin{equation} \label{eq:max_eq}
\delta(\veps) = \max \Bigg\{ \int\limits_{\mathbb{R}^d}  \max \{  f_X(t) - \ee^\veps f_{X'}(t) ,0  \} \, \dd t,  
\int\limits_{\mathbb{R}^d}  \max \{  f_{X'}(t) - \ee^\veps f_X(t) ,0  \} \, \dd t \Bigg\}.
\end{equation}
\end{lem}

As a direct corollary of Lemma~\ref{lem:tight_d}, we have the following.

\begin{lem}
	
The tight $\delta$ as a function of $\veps$ is given by
$$
\delta(\veps) = \max \left\{ \mathop{\mathbb{E}}\limits_{s \sim \omega_{X/{X'}}} 
\left[ \big( 1 - \ee^{\veps-s}  \big)_+\right], \mathop{\mathbb{E}}\limits_{s \sim \omega_{{X'}/X}}\left[ \big( 1 - \ee^{\veps-s}  \big)_+\right] \right\}.
$$
\begin{proof}

Let the privacy loss random variable $\omega_{X/{X'}}$ be defined as in Def.~\ref{defn:pld}. Then, with the change of variables
$s = \log \frac{f_X(t)}{f_{X'}(t)}$, we see that
\begin{equation*}
	\begin{aligned}
\int\limits_{\mathbb{R}^d}  \max \{  f_X(t) - \ee^\veps f_{X'}(t) ,0  \} \, \dd t 		
& = \int\limits_{\mathbb{R}^d}  \max \{  \big(1 - \ee^{\veps  - \log \frac{f_X(t)}{f_{X'}(t)}}\big) \cdot f_X(t) ,0  \} \, \dd t 		\\
& = \mathop{\mathbb{E}}\limits_{s \sim \omega_{X/{X'}}} \left[ \big( 1 - \ee^{\veps-s}  \big)_+\right].
	\end{aligned}
\end{equation*}
	
\end{proof}
\end{lem}

\subsubsection{Continuous Gaussian mechanism}

Recall: we focus on the case, where the underlying function is of the summative form
\begin{equation} \label{eq:f_sum}
	F(X,\theta) = \sum_{x \in X} g(x,\theta).
\end{equation}
And again, we denote $X'$ as a neighbouring dataset of $X$ such that $X' = X \cup \{ x' \}$. 

For the privacy loss distribution of the adaptive Gaussian mechanism, we have:

\begin{thm} \label{thm:adaptive_xyz}
Let $\widetilde{\omega}$ be the privacy loss random variable of the $k$-wise adaptive composition of a $d$-dimensional Gaussian mechanism,
where the underlying function $F$ is of the form \eqref{eq:f_sum}, $g(x,\theta)$ is Lipschitz-continuous as a function of $\theta$
for all $x$ and has a 2-norm exactly 1 for all $x$ and for all $\theta$.
Let $\omega$ be the $k$-wise non-adaptive composition of a $1$-dimensional Gaussian mechanism with sensitivity exactly 1. 
Then, for all inputs $\theta$ and all $S \subset \mathbb{R}$:
$$
\widetilde{\omega}(S,\theta) = \omega(S),
$$
i.e., the PLD of the multivariate adaptive composition is identical to that of the univariate non-adaptive composition.
\end{thm}
\begin{proof}
We consider a composition of two mechanisms, the general case can be shown using the same technique.
Let us assume first that $g(x,\theta)$ is everywhere differentiable as a function of $\theta$ for all $x$.


\textbf{(Translation)} We first make the change of variables
$$
\begin{bmatrix} 
	s_1  \\
	s_2 
	\end{bmatrix} =
	\begin{bmatrix} 
		t_1  \\
		t_2 
		\end{bmatrix}
		- 
		\begin{bmatrix} 
			\mathcal{M}_1(X,\theta) \\
			 \mathcal{M}_2(X,t_1)
	\end{bmatrix} =: F_1(t).
$$

Clearly $F_1$ is bijective and differentiable, with the inverse given by simple back-substitution: 
$$
F_1^{-1}(s) = \begin{bmatrix} 
		s_1  \\
		s_2 
		\end{bmatrix}
		+ 
		\begin{bmatrix} 
				\mathcal{M}_1(X,\theta) \\
				\mathcal{M}_2(X,s_1 + \mathcal{M}_1(X,\theta))
				\end{bmatrix}.
$$
We see that the Jacobian $\frac{\partial}{\partial s} F_1^{-1}(s)$ is a lower-triangular matrix with ones on the diagonal.
Thus $\mathrm{det} [ \frac{\partial}{\partial t} F_1^{-1}(t)] = 1$,
i.e., this change of variables preserves the measure. This is also easily seen in case of a $k$-wise composition, $k>2$. 
Now the privacy loss random variable expressed as
\begin{equation*}
	\begin{aligned}
		\widetilde{\omega}(S) = \int\limits_{\{ t \in \mathbb{R}^{2d} \, : \, \mathcal{L}_{X'/X}(t) \in S   \}} f_{X'}(t) \, \dd t 
			 	= \int\limits_{\{ t \in \mathbb{R}^{2d} \, : \, \widetilde{\mathcal{L}}_{X'/X}(t) \in S   \}} \widetilde{f}_{X'}(t) \, \dd t,
	\end{aligned}
\end{equation*}
where $\widetilde{f}_{X'}(t)$ denotes the density function of 
$$
\widetilde{\mathcal{M}}(X') = \big(   g(x',\theta) + \mathcal{N}(0,\sigma^2 I_d),   \widetilde{g}(x',t_1) + \mathcal{N}(0,\sigma^2 I_d) \big),
$$
where $t_1$ is the output of the first component, 
$\widetilde{g}(t_1) = g(t_1 + \mathcal{M}_1(X,\theta))$,
and $\widetilde{\mathcal{L}}_{X'/X}$ is determined by
$\widetilde{\mathcal{M}}(X')$ and $\widetilde{\mathcal{M}}(X)$, where 
$$
\widetilde{\mathcal{M}}(X) = \big(   \mathcal{N}(0,\sigma^2 I_d),  \mathcal{N}(0,\sigma^2 I_d) \big).
$$
 
\textbf{(Rotation)} Next, we make the change of variables
\begin{equation} \label{eq:change_of_vars}
	\begin{bmatrix} s_1 \\ s_2 \end{bmatrix} = \begin{bmatrix} U_1(\theta)^T & 0 \\ 0 & U_2(t_1)^T \end{bmatrix} 
		\begin{bmatrix} t_1 \\ t_2 \end{bmatrix} =: F_2(t),
\end{equation}
where $U_1(\theta), U_2(t_1) \in \mathbb{R}^{d \times d}$ are orthogonal matrices (i.e. $U_1^T U_1  = U_1 U_1^T = U_2^T U_2 = U_2 U_2^T = I_d$) such that $U_1$ is of the form
$$
U_1(\theta) = \begin{bmatrix} g(x',\theta) & \widetilde{U}_1(\theta) \end{bmatrix},
$$
where the columns of $\widetilde{U}_1(\theta) \in \mathbb{R}^{d \times (d-1)}$ give an orthonormal basis for the orthogonal complement of the subspace 
spanned by $g(x',\theta)$ such that $\widetilde{U}_1(\theta)$ depends continuously on $\theta$.
%
Similarly, $U_2$ is of the form
$$
U_2(t_1) = \begin{bmatrix}  \widetilde{g}(x',t_1) & \widetilde{U}_2(t_1) \end{bmatrix},
$$
where the columns of $\widetilde{U}_2(t_1) \in \mathbb{R}^{d \times (d-1)}$ give an orthonormal basis for the orthogonal complement of the subspace 
spanned by $\widetilde{g}(x',t_1)$ such that $\widetilde{U}_2(t_1)$ depends continuously on $t_1$.
%
Such basis matrices $\widetilde{U}_1(\theta)$ and $\widetilde{U}_2(t_1)$ clearly exist as $g(x',\theta)$ and $\widetilde{g}(x',t_1)$
depend continuously on $\theta$ and $t_1$, respectively.

Then, in particular, we have that
$$
U_1(\theta)^\mathrm{T} g(x',\theta) = e_1
\quad \textrm{and} \quad
U_2(t_1)^\mathrm{T}  \widetilde{g}(x',t_1) = e_1,
$$
where $e_1 = \begin{bmatrix} 1 & 0 & \ldots & 0 \end{bmatrix}^T \in \mathbb{R}^d$.
By simple back-substitution, we see that the inverse of the mapping $F_2$ is given by 
$$
	F_2^{-1}(s) = \begin{bmatrix} U_1(\theta) & 0 \\ 
		0 & U_2\big(U_1(\theta) s_1 \big)  \end{bmatrix} 
		\begin{bmatrix} s_1 \\ s_2 \end{bmatrix},
$$
and furthermore, its Jacobian is given by
$$
\frac{\partial}{\partial s} F_2^{-1}(s) = \begin{bmatrix} U_1(\theta) & 0 \\ 
 \frac{\partial}{\partial s_1} U_2\big(U_1(\theta) s_1 \big) s_2 & U_2\big(U_1(\theta) s_1 \big) \end{bmatrix}.
$$
Since the determinant of a block-triangular matrix is the product of the determinants of the matrices on the diagonal~\citep[pp.\;49]{horn2012matrix},
and since the absolute value of the determinant of an orthogonal matrix is 1, we see that for all $s = (s_1,s_2)$:
$$
\abs{ \mathrm{det} [\frac{\partial}{\partial s} F_2^{-1}(s)] } = \abs{ \mathrm{det} \big( U_1(\theta) \big)} \cdot \abs{ \mathrm{det} \big( U_2(U_1(\theta) s_1) \big)} = 1,
$$
i.e., also the rotation preserves the measure. In case of a $k$-wise composition, $k>2$, the Jacobian here will also be a lower-triangular matrix with orthogonal matrices
on its diagonal, from which $\abs{ \mathrm{det} [\frac{\partial}{\partial s} F_2^{-1}(s)] }=1$ follows.

After the rotation, the privacy loss can  be written as
\begin{equation*}
	\begin{aligned}
		\widetilde{\omega}(S) = \int\limits_{\{ t \in \mathbb{R}^{2d} \, : \, \widetilde{\mathcal{L}}_{X'/X}(t) \in S   \}} \widetilde{f}_{X'}(t) \, \dd t 
				= \int\limits_{\{ t \in \mathbb{R}^{2d} \, : \, \widehat{\mathcal{L}}_{X'/X}(t) \in S   \}} \widehat{f}_{X'}(t) \, \dd t,
	\end{aligned}
\end{equation*}
where $\widehat{f}_{X'}(t)$ denotes the density function of 
$$
\widehat{\mathcal{M}}(X') = \big(   e_1 + \mathcal{N}(0,\sigma^2 I_d),   e_1 + \mathcal{N}(0,\sigma^2 I_d) \big).
$$
and $\widehat{\mathcal{L}}_{X'/X}$ is determined by
$\widehat{\mathcal{M}}(X')$ and $\widehat{\mathcal{M}}(X)$, where 
$$
\widehat{\mathcal{M}}(X) = \big(   \mathcal{N}(0,\sigma^2 I_d),  \mathcal{N}(0,\sigma^2 I_d) \big).
$$
Thus, after the rotation we see that the privacy loss random variable of the adaptive composition is identical 
to that of a non-adaptive composition. As the coordinates $2$ to $d$ of $\widehat{\mathcal{M}}(X')$
and $\widehat{\mathcal{M}}(X)$ are identical, they do not contribute to the privacy loss.
More precisely: $\log \tfrac{f_{X',i}(t_i)}{f_{X,i}(t_i)} = 0$ for all $t_i$, $i \geq 2$, i.e. the coordinates $2$ to $d$ are simply integrated over $\mathbb{R}$ resulting in an integral containing only
the privacy loss of the first coordinate.

Thus, we see that for all $\theta$, 
$$
\widetilde{\omega}(S,\theta) = \omega(S),
$$
i.e., for all inputs $\theta$, the PLD of the multivariate adaptive composition is identical to that of the univariate non-adaptive composition.

Finally, we can loosen the assumption on differentiability of $g$ to Lipschitz-continuity of $g$, as follows.
By Rademacher's theorem, Lipschitz-bounded functions are almost everywhere differentiable~\cite[Thm.\;3.1.6,][]{federer1996}, 
and in case the transform $\phi$ is bi-Lipschitz-continuous but not necessarily everywhere differentiable, the change-of-variables formula 
$$
\int_U ( f \circ \phi) \abs{ \mathrm{det} \; \phi'(x)} \, \dd x = \int_{\phi(U)}f(x) \, \dd x
$$
still holds~\cite[see for example Thm.\;20.3 and its corollary 20.5 in][]{hewitt1965}. 
As $g(x,\theta)$ is Lipschitz-continuous, so is $F_1$ and also $U_1( \cdot )$ and $U_2(\cdot)$ and subsequently $F_2$ is also Lipschitz-continuous.
%
\end{proof}

\subsubsection{Monotonicity w.r.t. Sensitivity Using the Data Processing Inequality} \label{subsec:monotonicity}


We still need prove the fact that the condition $\norm{ g(x',\theta) }_2 = 1$ (for all $\theta$)
leads to an upper bound-$\delta(\veps)$-value for the cases where only a constraint $\norm{ g(x',\theta) }_2 \leq 1$ is imposed.
Repeating the steps above, we arrive at analysing 1-dimensional mechanisms
\begin{equation} \label{eq:c1c21}
	\widehat{\mathcal{M}}(X') = \big(   c_1(\theta) + \mathcal{N}(0,\sigma^2),   c_2(t_1) + \mathcal{N}(0,\sigma^2 I_d) \big),
\end{equation}
where $t_1$ denotes the output of the first component, $0\leq c_1(\theta) \leq 1$, $0 \leq c_2(t_1) \leq 1$ and
\begin{equation} \label{eq:c1c22}
\widehat{\mathcal{M}}(X) = \big(   \mathcal{N}(0,\sigma^2),  \mathcal{N}(0,\sigma^2) \big).
\end{equation}
By the post-processing property of DP, the analysis under the condition $\norm{ g(x',\theta) }_2 = 1$ (for all $\theta$) is equivalent
to considering the pair of mechanisms 
\begin{equation} \label{eq:c1c23}
	\widehat{\mathcal{M}}(X') = \big(   c_1(\theta) + \mathcal{N}(0,c_1(\theta)^2 \sigma^2),   c_2(t_1) + \mathcal{N}(0,c_2(t_1)^2 \sigma^2 I_d) \big),
\end{equation}
and
\begin{equation} \label{eq:c1c24}
\widehat{\mathcal{M}}(X) = \big(   \mathcal{N}(0,c_1(\theta)^2 \sigma^2),  \mathcal{N}(0,c_2(t_1)^2 \sigma^2) \big),
\end{equation}
where $t_1$ denotes the output of the first component. We see that we arrive to the pair of mechanisms
\eqref{eq:c1c21} and \eqref{eq:c1c22} by adding to the both \eqref{eq:c1c23} and \eqref{eq:c1c24} the
noise
$$
{Z} = \big( \mathcal{N}\big(0,(1-c_1(\theta)^2 ) \sigma^2 \big),  \mathcal{N}\big(0,(1-c_2(t_1)^2) \sigma^2 \big) \big).
$$
We know that the hockey-stick divergence that gives the tight $(\veps,\delta)$-bound 
is an $f$-divergence~\citep{barthe2013beyond}.
Using the data processing inequality for $f$-divergences~\citep[see e.g.][]{sason2016}, 
we see that by adding $Z$ to both $\widehat{\mathcal{M}}(X')$ and $\widehat{\mathcal{M}}(X)$ leads to and upper $(\veps,\delta)$-bound which shows the claim.
Using the data processing inequality is also motivated by the proof of~\cite[Thm.\;4]{mironov2019}, where the data processing 
property of R\'enyi divergences was used.

\subsubsection{Adaptive Composition of Multivariate Subsampled Gaussian Mechanisms}

The proof for the adaptive composition of subsampled Gaussian mechanism can be carried out by combining the proof of 
Theorems~\ref{thm:subsampled_xyz} and~\ref{thm:adaptive_xyz}. Under the condition $\norm{ g(x',\theta) }_2 = 1$ (for all $\theta$),
using rotation and translation as used in the proof of Thm.~\ref{thm:subsampled_xyz} and by showing that they 
preserve the measure (as in the proof of Thm.~\ref{thm:adaptive_xyz}), shows the claim.
The case $\norm{ g(x',\theta) }_2 \leq 1$ can be shown by noise-adding similarly to Subsection~\ref{subsec:monotonicity}.

\end{document}